\documentclass[10pt,journal]{IEEEtran}
\usepackage{amsmath,amssymb,amsthm}
\usepackage{graphicx}
\usepackage{booktabs}
\usepackage[hidelinks]{hyperref}
\usepackage{algorithm}
\usepackage{algpseudocode}

\newtheorem{theorem}{Theorem}
\newtheorem{proposition}{Proposition}
\newtheorem{definition}{Definition}

\newtheorem{lemma}{Lemma}
\theoremstyle{remark}

\newcommand{\RG}{\mathbf{R}_{G}}
\newcommand{\Tr}{\operatorname{Tr}}
\newcommand{\Deff}{D}
\newcommand{\Beff}{B_{\mathrm{eff}}}

\begin{document}

\title{Kernel Density Estimation by Spectral Decomposition: Data-Driven Tapering and Superposition
}

\author{Mitchell~A.~Thornton%
\thanks{The author is with the Darwin Deason Institute for Cyber Security and the
Department of Electrical and Computer Engineering, Southern Methodist University,
Dallas, TX 75275 USA
(e-mail: mitch@smu.edu).}}

\maketitle

\begin{abstract}
Kernel density estimation is among the most widely used tools in nonparametric statistics, and its
accuracy depends largely on one choice, the smoothing bandwidth. This paper treats bandwidth selection
and density estimation in the characteristic-function domain, where the cyclic group-averaged
covariance of the binned data has the squared empirical characteristic function as its spectrum: the
true characteristic function sits over a sampling-noise floor of $1/n$, and the bandwidth is the
spectral cutoff where the two meet. From this reading follow several methods. An automatic bandwidth
selector strips the floor and minimizes a frequency-domain error criterion, matching the rule of thumb
on smooth densities and approaching the best fixed bandwidth on multimodal ones, where the rule of thumb
over-smooths. An adaptive estimator generalizes the fixed kernel to the per-frequency optimal Wiener
taper, matching or surpassing the best fixed-bandwidth kernel on most standard densities, including
sharply peaked and comb-like cases on which fixed-bandwidth methods fail; deconvolution under known
measurement error follows in the same domain by dividing out the noise characteristic function. Because
the Wiener estimator resolves sharp structure but does not fit smooth bases as economically as a
mixture, a Gaussian mixture is combined with it two ways, a piecewise partition and a superposition of a
smooth base and a band-limited residual, of which the superposition is made the default. A data-driven
floor read from the spectrum replaces the assumed $1/n$ floor and stays robust on heaped and rounded
data, where aliasing lifts the true floor. On the Marron-Wand benchmark scored by exact integrated
squared error, the advantage emerges with sample size through a bias-variance tradeoff: the spectral
estimators carry low bias but pay in variance, so a corrected implementation of the Botev plug-in
bandwidth leads at $n=100$ while
the Wiener filter and the superposition take the top two ranks at $n=5000$. The methods are validated on
six real datasets, CRSP returns, NHANES self-reports, CMS dimuon and SDSS redshift spectra, a verified
random-beacon stream, and UNSW-NB15 traffic, and on a synthetic-data quality-checking use-case. All
claims are validated numerically and the experiments are reproducible.
\end{abstract}


\section{Introduction}
\IEEEPARstart{K}{ernel} density estimation (KDE) turns a sample into a smooth estimate of
its probability density and is one of the most heavily used methods in applied statistics
and machine learning \cite{silverman1986,scott2015,wandjones1995}. It underlies anomaly and
novelty detection, where low estimated density flags outliers, and in cybersecurity the same
principle drives behavioral baselining and zero-trust monitoring, where deviations from a
learned density of normal activity raise alerts. It supports generative and probabilistic
machine learning, where densities are sampled, compared, and used to score representations,
including emerging uses on latent and embedding spaces for novelty detection, dataset-quality
assessment, and large-language-model hallucination detection, settings in which the geometry
of the embedding space may itself carry symmetry that an algebraic estimator can exploit. In
quantitative finance and risk it builds the return and loss densities behind value-at-risk,
expected shortfall, and tail-risk analysis, where the distributions are routinely skewed,
heavy-tailed, and multimodal, a regime that motivates the attention this paper gives to
asymmetric and sharply peaked targets. In the life sciences it estimates gene-expression
distributions, cell-population structure, and protein-abundance profiles; in geographic
information systems it smooths incidence surfaces for crime analysis, traffic-accident
mapping, and disease-outbreak surveillance; and in autonomous systems and robotics it models
obstacle and pedestrian densities, vehicle-movement patterns, and the likelihood fields used
in localization. In signal processing it estimates spectrum-occupancy, emitter, interference,
and radar-target distributions; in computer vision it models color, texture, and object-location
distributions for tracking and for intensity-based image segmentation without assuming a
Gaussian mixture; and in the sciences it smooths source and clustering densities in astronomy,
firing-rate densities in neuroscience, and incidence surfaces in epidemiology. Across all of
these the estimate is only as good as its smoothing.

That smoothing is set by the bandwidth, and some practitioners consider its selection to be the most challenging issue in kernel density estimation \cite{wandjones1995,sheatherjones1991}. Too
small a bandwidth produces a spiky, high-variance estimate; too large a bandwidth blurs real
structure into a featureless mound. The standard rules of thumb assume a nearly Gaussian
target and over-smooth anything multimodal, while the more elaborate data-driven selectors
trade that bias for instability or for a delicate optimization.

This paper approaches the problem through algebraic diversity \cite{thornton2026algebraicdiversityprinciplesgrouptheoretic}, whose
central object is the group-averaged covariance, a projection of the sample covariance onto
the commutant of a group representation. Applied to the binned sample, that projection has
a familiar spectrum: its eigenvalues are the squared empirical characteristic function. A
kernel estimate is a low-pass filter on that same characteristic function, so the bandwidth
is a spectral cutoff, and selecting it is the spectral-support problem of finding where the
characteristic function descends into its $1/n$ sampling-noise floor. Two enhancements
follow. The first is an automatic bandwidth selector built from the noise-floor strip of the
algebraic framework. The second goes beyond bandwidth selection entirely: replacing the
fixed kernel with the per-frequency optimal taper yields an adaptive estimator that
outperforms the best fixed-bandwidth kernel.

The contributions are the following. The KDE bandwidth is identified with the spectral support of the
binned data, recasting selection as the problem of locating where the empirical characteristic function
meets its sampling-noise floor, and from this reading follow the bandwidth selector and the adaptive
Wiener estimator. The noise floor is itself a contribution: the algebraic residue, the part of the
spectrum that the matched group leaves unexplained, furnishes a data-driven maximum-entropy floor that
supersedes the fixed $1/n$ floor and is far more robust on heaped or rounded data, where aliasing lifts
the true floor above $1/n$. The Wiener estimator extends in the same domain to deconvolution under known
measurement error, dividing out the noise characteristic function with the residual-meets-floor rule
setting the cutoff. Because the Wiener filter resolves sharp structure but does not fit a smooth base as
economically as a mixture, a Gaussian mixture is combined with it two ways, a piecewise partition and a
superposition of a smooth base and a band-limited residual; the superposition decomposes densities
whose sharp and smooth parts overlap, which no single bandwidth resolves, and is made the default. Both
delivered densities, the band-limited adaptive estimate and the Gaussian mixture, are differentiable
almost everywhere, so model builders and digital-twin constructions can optimize against the recovered
density by gradient descent even when it fits sharp, spiky structure, and a synthetic-data
quality-checking use-case shows the method recovering contrived target densities exactly. On the
standard Marron-Wand benchmark of fifteen densities, scored by exact integrated squared error against
the known truth at three sample sizes, the advantage emerges with sample size through a bias-variance
tradeoff: the spectral estimators carry low bias but pay in variance, so at $n=100$ a corrected
implementation of the Botev plug-in
bandwidth leads and AD-Wiener alone is the robust spectral choice, while at $n=5000$ the Wiener filter
and the superposition take the top two ranks among seven, ahead of every classical baseline including
the corrected Botev diffusion bandwidth, the superposition surpassing a Gaussian mixture even where the
benchmark hands the mixture the correct model class. The reading is validated on six real datasets from distinct
domains: across thirteen CRSP series the spectral and mixture estimators recover the tail risk that a
Gaussian fit understates by more than two percent of a daily return; on NHANES survey self-reports the
data-driven floor stays robust where rounding defeats the fixed floor; on the CMS dimuon and SDSS
galaxy-redshift spectra, where sharp features carry probability mass on a smooth background, the
adaptive estimator improves held-out likelihood over the ordinary kernel estimate; on a NIST
randomness-beacon stream, verified uncorrupted by the statistical test suite, AD-Wiener recovers the
uniform target almost exactly where a Gaussian mixture cannot; and on UNSW-NB15 network traffic the
superposition gives the lowest held-out likelihood across flow features. All claims are validated
numerically and the experiments are reproducible.

\section{Background and Related Work}\label{sec:bg}
\subsection{Kernel density estimation}
Given samples $x_1,\dots,x_n$, the kernel density estimate with kernel $K$ and bandwidth
$h$ is $\hat f_h(x)=\tfrac{1}{n}\sum_{j} K_h(x-x_j)$, $K_h(u)=h^{-1}K(u/h)$. Its mean
integrated squared error decomposes into a squared bias that grows with $h$ and a variance
that falls with $h$, so a single scalar trades the two off \cite{wandjones1995}. The choice
of kernel matters little; the choice of $h$ is decisive.

\subsection{Bandwidth selection}
The simplest selectors are reference rules that set $h$ from the sample spread under a
Gaussian assumption, such as Silverman's and Scott's rules \cite{silverman1986,scott2015}.
These are fast but over-smooth structured densities. Data-driven selectors improve on them:
least-squares and likelihood cross-validation; the plug-in method of Sheather and Jones,
long the reference standard \cite{sheatherjones1991}; the diffusion estimator of Botev and
colleagues, an FFT-based selector that also yields an improved Sheather-Jones bandwidth
\cite{botev2010}; and the empirical-characteristic-function selectors of Chiu
\cite{chiu1991}, which truncate the characteristic function at the frequency where it meets
its sampling floor. Bandwidth selection remains the aspect of kernel density estimation that
most determines accuracy and the one most studied.

\subsection{Frequency-domain and optimal-kernel estimators}
A line of work central to what follows treats density estimation directly in the frequency
domain. Watson and Leadbetter derived the linear filter on the empirical characteristic
function that minimizes the mean integrated squared error \cite{watsonleadbetter1963}, a
classical result this paper builds on rather than claims. Realizations differ in how the
unknown characteristic function is estimated: Chiu's selectors \cite{chiu1991} truncate at the
sampling floor and retain a kernel taper, while the self-consistent estimator of Bernacchia
and Pigolotti \cite{bernacchia2011} solves for the optimal filter as a fixed point of the
empirical characteristic function, admitting frequencies whose squared amplitude clears a
fixed threshold of approximately $4/n$. Linear frequency-domain estimates need not be
densities; the correction of Glad, Hjort, and Ushakov \cite{glad2003} repairs this without
increasing the integrated squared error. The contribution here to this line is the one-shot
floor strip, the data-driven noise floor read from the residual spectrum, which replaces the
fixed $1/n$ or $4/n$ thresholds and survives discretized data
(Section~\ref{sec:select}), and the superposition estimator of
Section~\ref{sec:superpose}.

\subsection{Other enhancements}
Beyond the global bandwidth, several lines of work enhance the estimator itself. Variable
and adaptive bandwidth methods let the smoothing vary with local density, narrowing in dense
regions and widening in sparse tails \cite{abramson1982,breiman1977}. Boundary-corrected
estimators remove the bias that ordinary KDE suffers near the edge of a bounded support, by
reflection or by boundary kernels \cite{schuster1985,jones1993}. Transformation methods map
the data to a scale where a single bandwidth suffices and then map back \cite{wmr1991}.
Deconvolution estimators recover a density from noisy observations by dividing out the noise
characteristic function \cite{stefanski1990}. The present work adds to this literature a
characteristic-function reading that yields both a bandwidth selector and, going beyond a fixed
kernel, an optimal adaptive estimator, and it extends in the same domain to deconvolution under known
measurement error (Section~\ref{sec:deconv}).

\section{The Group-Averaged Spectrum and the Empirical Characteristic Function}\label{sec:ecf}
Let the samples be binned into a histogram $p\in\mathbb{R}^{M}$ with $\sum_m p_m=1$ on a
uniform grid of spacing $\Delta x$, viewed as a signal on the cyclic index set. Let $S$ be
the unitary cyclic shift and $G=\mathbb{Z}_M$ the group it generates. The single-observation
group-averaged covariance \cite{thornton2026algebraicdiversityprinciplesgrouptheoretic} is
$\RG=\tfrac{1}{M}\sum_{g\in G}(S^{g}p)(S^{g}p)^{H}$, the Reynolds projection of $pp^{H}$
onto the commutant of the cyclic representation. Write $F$ for the unitary DFT and
$\hat\varphi(t_k)=(Fp)_k$ for the empirical characteristic function at
$t_k=2\pi k/(M\Delta x)$.

\begin{theorem}[Spectrum identity]\label{thm:identity}
$\RG$ is circulant, is diagonalized by $F$, and its eigenvalues are the squared empirical characteristic function (ECF):
$(F\RG F^{H})_{kl}=\delta_{kl}\,\lvert\hat\varphi(t_k)\rvert^{2}$.
\end{theorem}
\begin{proof}
The shift theorem gives $FS^{g}F^{H}=D^{g}$ with
$D=\operatorname{diag}(e^{-\jmath 2\pi k/M})$, so
$F\RG F^{H}=\tfrac{1}{M}\sum_{g}D^{g}(Fp)(Fp)^{H}D^{-g}$. The $(k,l)$ entry of $(Fp)(Fp)^H$
is $(Fp)_k(Fp)_l^{*}$ and conjugation by $D^g$ multiplies it by $e^{-\jmath2\pi(k-l)g/M}$;
character orthogonality $\tfrac{1}{M}\sum_g e^{-\jmath2\pi(k-l)g/M}=\delta_{kl}$ leaves
$\delta_{kl}\lvert(Fp)_k\rvert^2$. Diagonal in the Fourier basis means circulant
\cite{gray2006}.
\end{proof}

The binned data's algebraic spectrum is therefore the squared ECF, the same kind of object
the bandwidth-estimation framework treats: a true characteristic function near the origin
over a sampling-noise floor.

\begin{definition}[Effective dimension]\label{def:D}
With eigenvalues $\lambda_k=\lvert\hat\varphi(t_k)\rvert^{2}$, the participation ratio is
$\Deff=(\sum_k\lambda_k)^2/\sum_k\lambda_k^2=(\Tr\RG)^2/\lVert\RG\rVert_F^2$, the effective
number of occupied spectral modes \cite{hill1973}.
\end{definition}

\begin{proposition}[Noise floor]\label{prop:floor}
$\mathbb{E}\lvert\hat\varphi(t)\rvert^{2}=\lvert\varphi(t)\rvert^{2}
+(1-\lvert\varphi(t)\rvert^{2})/n\to 1/n$ as $\varphi(t)\to 0$.
\end{proposition}
\begin{proof}
$\mathbb{E}\hat\varphi(t)=\varphi(t)$ and
$\operatorname{Var}\hat\varphi(t)=\tfrac{1}{n}(1-\lvert\varphi(t)\rvert^2)$, so
$\mathbb{E}\lvert\hat\varphi\rvert^2=\lvert\varphi\rvert^2+\tfrac1n(1-\lvert\varphi\rvert^2)$.
\end{proof}

\section{Spectral Bandwidth Selection by the Algebraic Residue}\label{sec:select}
A Gaussian kernel estimate has Fourier transform $\hat\varphi(t)\psi(ht)$ with
$\psi(s)=e^{-s^2/2}$, so the bandwidth $h$ is a low-pass cutoff on the ECF, and choosing it
is choosing where to stop trusting $\hat\varphi$. The algebraic framework supplies the rule:
strip the known floor, $\hat S_k=\max(\lvert\hat\varphi(t_k)\rvert^2-1/n,0)$, truncated
beyond the first frequency at which the smoothed squared ECF drops into the floor (this
removes the rectification residual that otherwise biases the high-frequency tail), and select
\begin{equation}\label{eq:Gh}
h_{\mathrm{AD}}=\arg\min_{h}\ \sum_k\lvert\hat\varphi(t_k)\rvert^2\psi(ht_k)^2
-2\sum_k\hat S_k\,\psi(ht_k),
\end{equation}
which estimates the integrated squared error up to a constant. This is the
empirical-characteristic-function selector of \cite{chiu1991} read through the floor strip
and effective-support cutoff of the present framework.

\begin{proposition}[Inflation and the strip]\label{prop:inflate}
For an expected spectrum with $K$ in-band bins at level $a=S+b$ and $M-K$ out-of-band bins
at the floor $b$, the participation ratio of the expected spectrum is
$\bar\Deff=(Ka+(M-K)b)^2/(Ka^2+(M-K)b^2)$, which tends to $K$ as $b/a\to0$ and to $M$ as
$a\to b$. Subtracting a consistent estimate of $b$ and clipping restores $\bar\Deff\to K$.
\end{proposition}
\begin{proof}
The expression is Definition~\ref{def:D} on the two-level spectrum; the limits are
immediate, and after subtraction the out-of-band level is zero, leaving $K$ equal nonzero
values.
\end{proof}

\subsection{The algebraic-residue floor option}\label{sec:residue}
The strip above assumes the floor is exactly $1/n$. That holds for an exact i.i.d.\ sample,
but is violated whenever the data are rounded, heaped, or otherwise discretized, conditions
that are common in practice and that lift the high-frequency floor above $1/n$ through
aliasing. The algebraic-diversity framework supplies a data-driven alternative through the
algebraic residue \cite{thornton2026algebraicdiversityprinciplesgrouptheoretic}: the component that the best-matched group
cannot concentrate is the maximum-entropy, white part of the group-averaged covariance, and
its level is read directly from the spectrum rather than assumed. For the squared ECF the
noise bins are exponential, so the robust floor estimate is the order statistic
$\hat b=\operatorname{median}_k\lvert\hat\varphi(t_k)\rvert^2/\ln 2$, the level at which the
spectrum is indistinguishable from white. The structured part is then recovered not by a
hard cut but by the soft Wiener (LMMSE) gain
\begin{equation}\label{eq:softgain}
w_k=\frac{(\lvert\hat\varphi(t_k)\rvert^2-\hat b)_+}{\lvert\hat\varphi(t_k)\rvert^2},
\qquad \hat S_k=w_k\,\lvert\hat\varphi(t_k)\rvert^2,
\end{equation}
which tapers the floor rather than truncating it and so reduces the residual structure left
by a hard cut. The hard strip of \eqref{eq:Gh} is the binary limit
$w_k\in\{0,1\}$ of this same gain. Because the matched group of the binned measure is known
to be cyclic, this requires only the floor-and-gain primitive of the residue construction,
not the group-search or deflation machinery of the general method. We keep the $1/n$ strip
as the default and expose the residue strip as an option; Section~\ref{sec:exp} shows the two
are equivalent on exact samples but that the residue strip is far more robust under heaping.

\section{Adaptive Wiener Density Estimation}\label{sec:wiener}
The kernel taper $\psi(ht)$ is a one-parameter family; nothing requires the optimal
frequency-domain filter to lie in it. Write a general linear estimator as the inverse
transform of $g(t)\hat\varphi(t)$ for a real filter $g$.

\begin{proposition}[Optimal linear filter, Watson and Leadbetter \cite{watsonleadbetter1963}]\label{prop:wiener}
The filter minimizing the mean integrated squared error is the Wiener filter
\begin{equation}\label{eq:W}
g^{\star}(t)=\frac{\lvert\varphi(t)\rvert^{2}}
{\lvert\varphi(t)\rvert^{2}+(1-\lvert\varphi(t)\rvert^{2})/n},
\end{equation}
and the fixed-bandwidth kernel estimator is the restriction $g=\psi(ht)$ of this filter to a
one-parameter family.
\end{proposition}
\begin{proof}
By Parseval the mean integrated squared error is
$\tfrac{1}{2\pi}\int \mathbb{E}\lvert g\hat\varphi-\varphi\rvert^2\,dt$. Since
$\mathbb{E}\hat\varphi=\varphi$ and
$\operatorname{Var}\hat\varphi=(1-\lvert\varphi\rvert^2)/n$, the integrand is
$(g-1)^2\lvert\varphi\rvert^2+g^2(1-\lvert\varphi\rvert^2)/n$. Differentiating in $g$ and
setting to zero gives \eqref{eq:W}. Choosing $g=\psi(ht)$ recovers the kernel estimator.
\end{proof}
This is the classical optimal-kernel result of Watson and Leadbetter
\cite{watsonleadbetter1963}; the short proof is reproduced to fix notation and to make the
kernel-restriction reading explicit.

\begin{definition}[AD-Wiener estimator]\label{def:wiener}
Estimate $\lvert\varphi(t)\rvert^2$ by the stabilized, floor-stripped squared ECF
$\hat S(t)$ and form $\hat g(t)=\hat S(t)/(\hat S(t)+1/n)$. The AD-Wiener density estimate
is the inverse transform of $\hat g(t)\hat\varphi(t)$, clipped to be nonnegative and
renormalized; this correction of a linear estimate that is not a density does not increase
the integrated squared error \cite{glad2003}.
\end{definition}

The AD-Wiener estimator is thus a one-shot, floor-stripped plug-in realization of the
Watson-Leadbetter filter. It differs from the self-consistent realization of Bernacchia and
Pigolotti \cite{bernacchia2011} in the estimate of the unknown $\lvert\varphi\rvert^2$, a
single strip of the known sampling floor rather than a fixed-point iteration, and in the
floor itself, which the residue strip of Section~\ref{sec:select} measures from the data
rather than fixing at the exact-sampling level; the measured floor is what carries the method
through the rounded and heaped data of Section~\ref{sec:heaping}, where every fixed threshold
fails.

Because $\hat g$ adapts the smoothing per frequency, keeping high-frequency content exactly
where the data support it and discarding it elsewhere, the estimator can resolve structure
at several scales at once, which no single bandwidth can. Proposition~\ref{prop:wiener}
explains why it can match or surpass the best fixed bandwidth: it optimizes over all
per-frequency tapers, of which the kernel family is a slice. The residue floor of
Section~\ref{sec:residue} substitutes directly here: replacing $1/n$ by $\hat b$ and using
the soft gain \eqref{eq:softgain} is exactly the AD-Wiener filter with a data-driven floor,
and is the form we recommend for discretized data.

\section{Experiments}\label{sec:exp}
The selector \eqref{eq:Gh} and the AD-Wiener estimator of Definition~\ref{def:wiener} are
benchmarked on the Marron-Wand test densities \cite{marronwand1992}, including the extreme
multi-scale cases, against Silverman's rule, the diffusion estimator of \cite{botev2010}
applied as an estimator at its improved Sheather-Jones bandwidth, the empirical-
characteristic-function plug-in selector of \cite{chiu1991}, and the adaptive variable-
bandwidth estimator of \cite{abramson1982} with a square-root pilot law. Error is mean
integrated squared error over replications against the known truth; for the fixed-bandwidth
methods a common Gaussian kernel is used so that only the bandwidth differs, and ``best
fixed'' is the ISE-minimizing single bandwidth, the strongest a fixed kernel can do. The AD
methods here use the default $1/n$ strip. Seeds are recorded in \texttt{DATA.md}.

\begin{table*}[t]
\centering
\caption{Mean integrated squared error ($\times10^{3}$) on exact samples. ``best fixed'' is
the ISE-minimizing single bandwidth (Gaussian kernel); the best data-driven method is in \textbf{bold}. Botev is the diffusion estimator at the improved Sheather-Jones bandwidth
\cite{botev2010}, Chiu the ECF plug-in selector \cite{chiu1991}, Abram. the adaptive
variable-bandwidth estimator \cite{abramson1982}.}
\label{tab:kde}
\setlength{\tabcolsep}{6pt}
\begin{tabular}{lr r rrrr rr}
\toprule
density & $n$ & best fix & Silver. & Botev & Chiu & Abram. & AD-bw & AD-Wien. \\
\midrule
Gaussian        & 200  & 2.77 & 3.55 & \textbf{3.33} & 3.66 & 4.71 & 4.00 & 4.49 \\
Gaussian        & 2000 & 0.52 & 0.62 & 0.61 & 0.58 & 0.84 & 0.60 & \textbf{0.56} \\
Bimodal         & 200  & 4.19 & 4.63 & 5.05 & 5.19 & \textbf{4.34} & 5.55 & 5.46 \\
Bimodal         & 2000 & 0.88 & 1.03 & 0.94 & 0.92 & 0.82 & 0.99 & \textbf{0.81} \\
Kurtotic        & 200  & 23.32 & 65.59 & 26.98 & \textbf{24.96} & 25.68 & 25.24 & 25.31 \\
Kurtotic        & 2000 & 4.31 & 20.48 & 4.68 & 4.57 & \textbf{3.30} & 4.55 & 3.79 \\
Claw            & 200  & 23.60 & 48.53 & 31.59 & 42.20 & 49.05 & 32.42 & \textbf{29.56} \\
Claw            & 2000 & 3.73 & 33.93 & 3.91 & 3.87 & 27.82 & 3.93 & \textbf{3.01} \\
Asym. claw      & 200  & 15.38 & 22.79 & 19.68 & 18.42 & 22.52 & \textbf{18.33} & 19.35 \\
Asym. claw      & 2000 & 3.88 & 13.10 & 4.52 & 4.24 & 11.38 & \textbf{4.16} & 4.54 \\
Smooth comb     & 200  & 25.27 & 68.17 & 28.79 & \textbf{27.09} & 63.80 & 27.28 & 28.47 \\
Smooth comb     & 2000 & 6.59 & 43.23 & 7.45 & 6.96 & 40.18 & \textbf{6.84} & 7.27 \\
Discrete comb   & 200  & 26.52 & 92.98 & 29.17 & \textbf{28.47} & 88.83 & 29.75 & 29.04 \\
Discrete comb   & 2000 & 5.06 & 50.33 & 5.28 & 8.35 & 42.20 & 5.16 & \textbf{4.43} \\
Strongly skewed & 200  & 25.14 & 91.71 & 29.91 & \textbf{28.74} & 66.33 & 29.47 & 31.51 \\
Strongly skewed & 2000 & 4.84 & 53.26 & \textbf{5.15} & 5.21 & 29.91 & 5.16 & 5.18 \\
\bottomrule
\end{tabular}
\end{table*}

Table~\ref{tab:kde} places the algebraic methods against the nearest competitors rather than
the rule of thumb alone. One implementation note applies to every Botev column and curve in
this paper: the improved Sheather-Jones bandwidth is computed by a direct implementation of the
Botev-Grotowski-Kroese fixed point \cite{botev2010}. A widely used library implementation was
found to undersmooth systematically, returning a bandwidth about $2.6$ times below the fixed
point's value on smooth targets (e.g.\ $0.075$ against $0.198$ on a standard normal at
$n=5000$, where the asymptotically optimal value is $0.193$), which misrepresents the method;
the first posted version of this paper carried results produced with it, and they are
regenerated here. On the smooth unimodal densities every reasonable method sits near
the best fixed bandwidth, with Botev, Silverman, or Abramson better at the smaller sample and
AD-Wiener better at the larger, where it edges even the best fixed bandwidth. The two strongest
classical competitors are the spectral ones, and they are genuinely close: Chiu's ECF selector,
which shares the spectral viewpoint developed here, is the best data-driven method on several
densities at $n=200$, and the corrected Botev selector is best on the Gaussian at $n=200$ and
the strongly skewed density at $n=2000$, and within a few percent of the leader on most of the
rest. The adaptive estimator is the best method on the kurtotic spike at $n=2000$, where a
single variable bandwidth suits one sharp peak, but it over-adapts and fails on the claw and
the combs. The spectral methods remain the most consistent performers on the genuinely
multi-scale cases: the AD-Wiener estimator is best on the claw at both sample sizes and the
discrete comb at $n=2000$, on the claw at $n=2000$ surpassing the best fixed bandwidth itself
by about twenty percent because no single bandwidth can both resolve the narrow peaks and
smooth the broad base, while the fixed-kernel AD selector leads on the asymmetric claw and the
smooth comb.

\begin{table}[t]
\centering
\caption{Mean integrated squared error ($\times10^{3}$) on heaped data (samples rounded to
$0.1$, $n=2000$). The fixed $1/n$ floor fails as aliasing lifts the true floor; the residue
floor adapts. The simple-strip Wiener estimate is destroyed by the rounding spectrum, the
residue strip is not.}
\label{tab:heaped}
\setlength{\tabcolsep}{5pt}
\begin{tabular}{l rr rr}
\toprule
 & \multicolumn{2}{c}{AD-bw} & \multicolumn{2}{c}{AD-Wiener} \\
\cmidrule(lr){2-3}\cmidrule(lr){4-5}
density & simple & residue & simple & residue \\
\midrule
Gaussian        & 0.70 & \textbf{0.60} & 0.65 & \textbf{0.44} \\
Bimodal         & 1.01 & \textbf{0.93} & 0.88 & \textbf{0.73} \\
Kurtotic        & 151.7 & \textbf{91.0} & 732.4 & \textbf{87.7} \\
Claw            & \textbf{35.1} & 43.2 & 376.0 & \textbf{45.6} \\
Asym. claw      & 85.1 & \textbf{13.7} & 2452.9 & \textbf{14.1} \\
Smooth comb     & 39.6 & \textbf{20.6} & 3874.1 & \textbf{22.5} \\
Discrete comb   & 23.1 & \textbf{19.7} & 2088.6 & \textbf{20.5} \\
Strongly skewed & 257.7 & \textbf{52.3} & 4712.3 & \textbf{62.5} \\
\bottomrule
\end{tabular}
\end{table}

On exact samples the residue strip of Section~\ref{sec:residue} reproduces the default to
within Monte-Carlo error, so Table~\ref{tab:kde} is unchanged by the choice and is not
repeated. The two diverge once the floor assumption is violated. Table~\ref{tab:heaped}
rounds each sample to one decimal, a routine real-data condition, which aliases mass into the
high frequencies and lifts the true floor well above $1/n$. The fixed-floor estimator then
places its cutoff wrongly: for the bandwidth selector this inflates the error several-fold,
and for the Wiener estimator it is catastrophic, since the unsuppressed rounding spectrum
passes through the filter and the inverse transform is overwhelmed by spurious high-frequency
oscillation, with the error rising by one to three orders of magnitude. The residue floor,
read from the spectrum rather than assumed, tracks the elevated level and keeps both
estimators stable. Figure~\ref{fig:heaped} shows the effect directly.

\begin{figure*}[t]
\centering
\includegraphics[width=\textwidth]{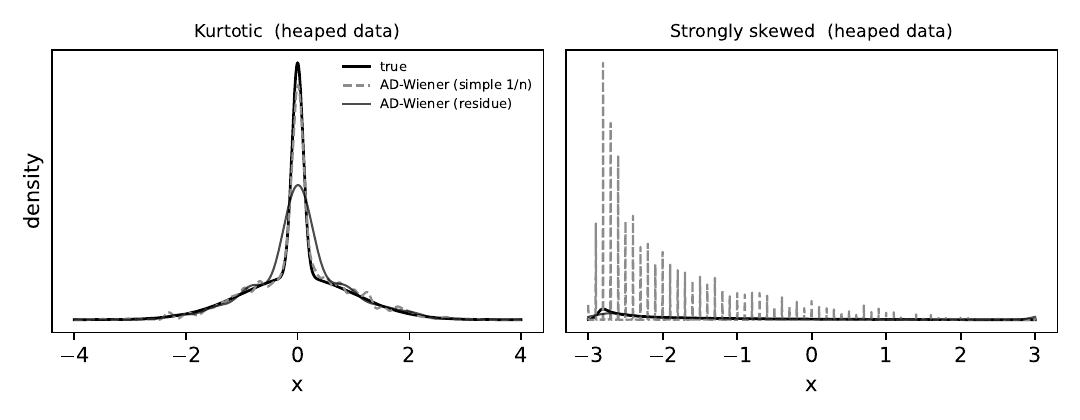}
\caption{Heaped data (rounded to $0.1$). The simple $1/n$ Wiener estimate (dashed) throws
spurious spikes at the rounding locations; the residue-floor Wiener estimate (solid) tracks
the truth.}
\label{fig:heaped}
\end{figure*}

Figs.~\ref{fig:kde} and~\ref{fig:hard} show why. The rule of thumb fuses the claw and comb
peaks into a single mode and misses the strongly skewed spike. The corrected Botev estimate
resolves the claw modes cleanly at this sample size, and its remaining limitation is the one
any single bandwidth has: on the discrete comb it must round the razor-thin peaks to keep the
wide ones smooth. The AD-Wiener estimate tracks the truth across scales, resolving a sharp
spike over a broad base (kurtotic), peaks of decreasing width (asymmetric claw), and a mixture
of wide and razor-thin peaks (discrete comb) at once.

\begin{figure*}[t]
\centering
\includegraphics[width=\textwidth]{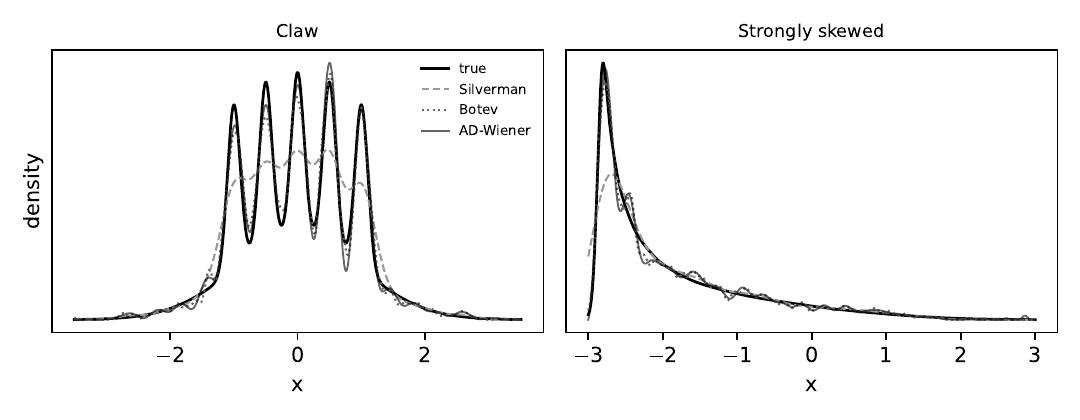}
\caption{Estimates at $n=2000$. The adaptive AD-Wiener estimate resolves the structure that
the rule of thumb smooths away and tracks the truth alongside the corrected Botev estimate.}
\label{fig:kde}
\end{figure*}

\begin{figure*}[t]
\centering
\includegraphics[width=\textwidth]{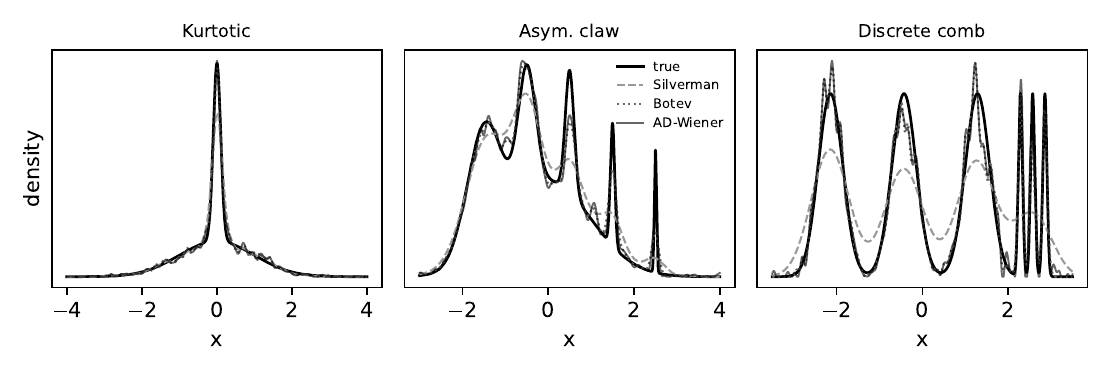}
\caption{Extreme multi-scale densities at $n=2000$: a sharp spike over a broad base, peaks
of decreasing width, and wide together with razor-thin peaks. Fixed-bandwidth methods cannot
serve all scales at once; the per-frequency AD-Wiener filter can.}
\label{fig:hard}
\end{figure*}

\section{Generality of the Enabling Estimator}\label{sec:general}
The bandwidth estimator that enables the above is not special to densities. The same
group-averaged spectrum estimates the occupied bandwidth of a signal, where
$\Beff=\Deff\,f_s/M$ with $f_s$ the sampling rate. It is exact for a flat band
(Theorem~\ref{thm:flat}); for a single look at a stochastic process it reports an effective
width of about half the support, a property of the single-look periodogram that averaging
removes (Propositions~\ref{prop:half} and~\ref{prop:consistency}); and the same floor strip
that serves the density selector controls noise inflation. Validation of the signal-bandwidth
estimator across signal-to-noise ratio, a comparison with the conventional occupied
bandwidth, and an extension to chirps through the metaplectic chirp rate are summarized in
Figs.~\ref{fig:mech}--\ref{fig:b99}; the same noise-floor strip and effective-support
cutoff drive both applications.

\begin{theorem}[Flat-band exactness]\label{thm:flat}
If the spectrum occupies $K$ bins with equal power, $\Deff=K$ and $\Beff=K\,f_s/M$, the true
occupied bandwidth.
\end{theorem}
\begin{proof}
With $K$ equal eigenvalues $c$ and the rest zero,
$\Deff=(Kc)^2/(Kc^2)=K$.
\end{proof}

\begin{proposition}[Single stochastic look]\label{prop:half}
If the in-band squared-ECF values are independent and exponential, as for one observation of
a complex Gaussian band-limited process on $K$ bins, then
$\mathbb{E}[\Deff]\approx(K+1)/2$.
\end{proposition}
\begin{proof}
For independent exponential $\lambda_k$ with mean $\mu$ and $\mathbb{E}\lambda_k^2=2\mu^2$,
$\mathbb{E}(\sum\lambda)^2=K\mu^2+K^2\mu^2$ and $\mathbb{E}\sum\lambda^2=2K\mu^2$; the ratio
of expectations is $(K+1)/2$.
\end{proof}

\begin{proposition}[Consistency under averaging]\label{prop:consistency}
The periodogram averaged over $L$ independent observations converges to the power spectral
density, and its participation ratio converges to that of the spectrum; for a flat band
$\Deff\to K$.
\end{proposition}
\begin{proof}
The strong law gives bin-wise convergence; $\Deff$ is continuous in the eigenvalues away
from the origin.
\end{proof}

\begin{figure*}[t]
\centering
\includegraphics[width=\textwidth]{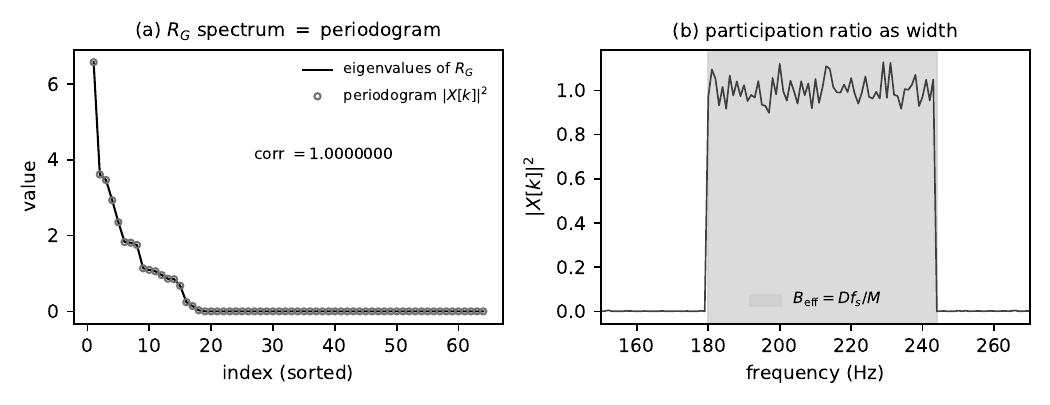}
\caption{The eigenvalues of the group-averaged covariance equal the periodogram; the
participation ratio is the effective occupied width.}
\label{fig:mech}
\end{figure*}
\begin{figure*}[t]
\centering
\includegraphics[width=\textwidth]{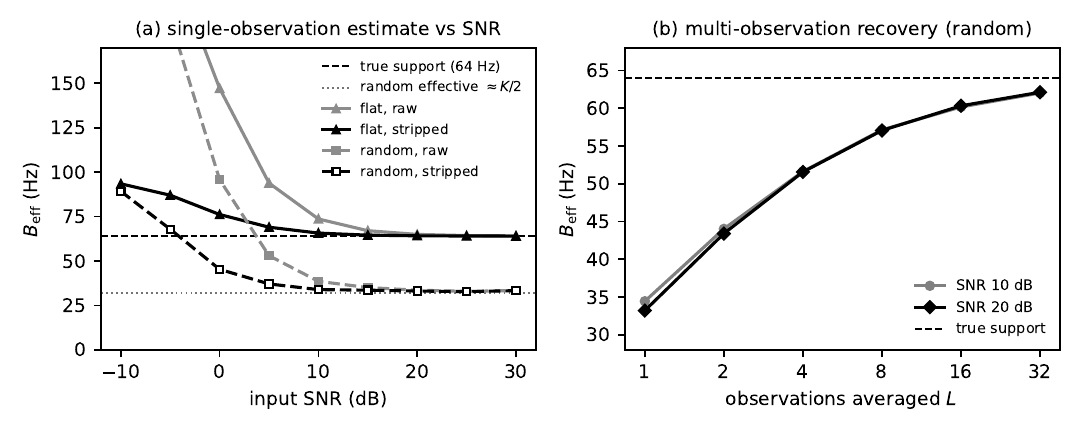}
\caption{Single-observation bandwidth versus signal-to-noise ratio, and recovery of the
support bandwidth by averaging.}
\label{fig:sweep}
\end{figure*}
\begin{figure*}[t]
\centering
\includegraphics[width=\textwidth]{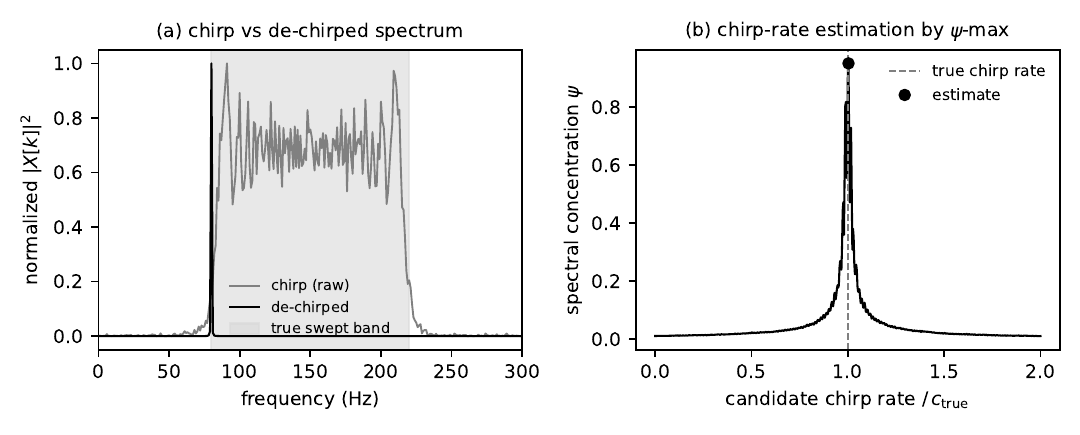}
\caption{For a chirp the cyclic estimate measures the swept bandwidth; the metaplectic chirp
rate, recovered by spectral-concentration maximization, de-chirps to the instantaneous
bandwidth.}
\label{fig:chirp}
\end{figure*}
\begin{figure*}[t]
\centering
\includegraphics[width=\textwidth]{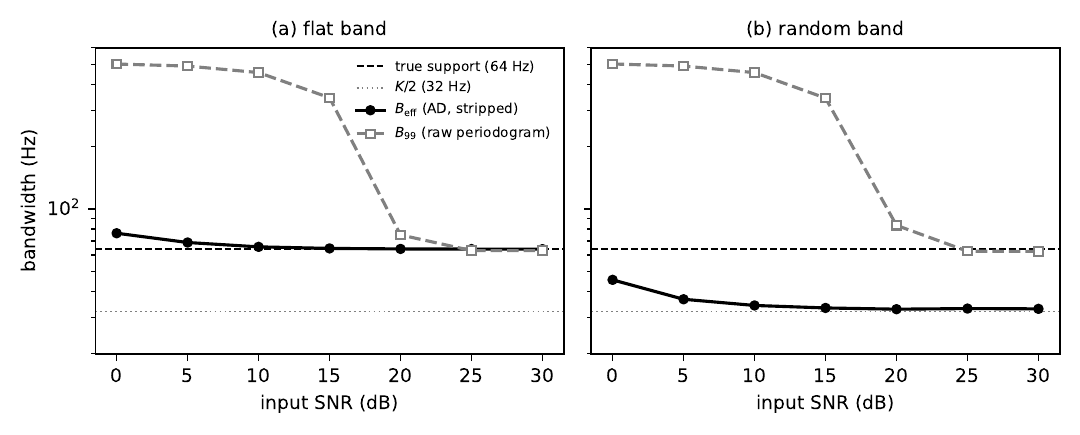}
\caption{The algebraic effective bandwidth against the conventional ninety-nine percent
occupied bandwidth: the former measures an effective width and is far more robust to noise.}
\label{fig:b99}
\end{figure*}

\section{Recognizing and Handling Asymmetric and Heavy-Tailed Densities}\label{sec:hard}
The densities on which a fixed kernel struggles most are the sharply peaked and the strongly
asymmetric, where one global smoothing cannot serve a narrow peak and a broad tail at once.
The algebraic reading both detects this regime and suggests a remedy. As a detector, the
effective dimension $D_2$ of the stripped coherent spectrum (Definition~\ref{def:D}) is
elevated for spike-plus-tail densities: about $5.9$ for the kurtotic density and $4.3$ for the
strongly skewed one, against $2.2$ to $2.9$ for the smooth, bimodal, and comb densities. A
classical skewness or kurtosis statistic carries the same information and separates the
asymmetric case from the symmetric spike.

The remedy follows from the matched-group viewpoint. A fixed grid makes the translation
(cyclic) group the matched symmetry, which suits stationary, near-symmetric targets; a
strongly asymmetric density violates that stationarity. Restoring it with a monotone
symmetrizing transform \cite{wmr1991}, estimating with the AD-Wiener filter in the
transformed domain, and mapping back through the Jacobian recovers the match. On the strongly
skewed density this lowers the mean integrated squared error from $5.31$ to $3.88$
($\times10^{3}$), about a twenty-seven percent reduction, consistent across sample sizes and
seeds. The transform is neutral on the symmetric kurtotic spike, which a monotone warp cannot
symmetrize, and it distorts the regular structure of the comb densities, so it is applied only
when the skewness probe fires; the comb densities, whose coherent effective dimension is low,
do not trigger it. This is the robustness the residue treatment was built for, applied here to
the geometry of the sample rather than to the noise floor.

The construction is not confined to the spectral paradigm. The translation group is only one
choice of matched symmetry, and the same group-averaging can be carried out under any group the
data respect. Under the dilation (scaling) group, for instance, the eigenbasis is a wavelet
multiresolution rather than the Fourier basis, and the coherent-versus-residue split of
Section~\ref{sec:residue} becomes wavelet-coefficient thresholding \cite{mallat1989,donoho1996}.
We tested a translation-invariant realization of this dilation-group estimator and found that on
the smooth Gaussian-mixture densities considered here it does not improve on the spectral
construction, its advantage being confined to densities with genuine discontinuities; the
spectral construction is therefore retained, with the matched-group viewpoint left as the
organizing principle.

\section{Deconvolution under Known Measurement Error}\label{sec:deconv}
Suppose the observations are corrupted by measurement error, $Y = X + \varepsilon$, with
$\varepsilon$ of known distribution independent of $X$; the target is the density of $X$, not of
$Y$. The observed density is the convolution $f_Y = f_X * f_\varepsilon$, so an ordinary kernel
estimator faithfully recovers the blurred $f_Y$ and is biased for $f_X$. In the
characteristic-function domain the convolution is a product, $\varphi_Y = \varphi_X\,
\varphi_\varepsilon$, so deconvolution is division by the known noise characteristic function,
carried out in the same domain in which the bandwidth selector and the AD-Wiener filter already
operate \cite{stefanski1990,fan1991}. This is the spectral form of the Wiener filtering tradition
the adaptive estimator descends from.

Division by $\varphi_\varepsilon$ amplifies the floor. The flat $1/n$ sampling floor on the
empirical characteristic function becomes the frequency-shaped floor
$(1/n)/|\varphi_\varepsilon(t)|^2$, which blows up where $\varphi_\varepsilon$ decays, and this is
the whole difficulty of deconvolution: the floor grows fastest exactly where one most wants to
divide. The AD-Wiener taper carries over unchanged except that it is now applied against this
amplified floor; the cutoff is the frequency at which the deconvolved power descends to meet the
floor, found by comparing the empirical residual to the computed one and stopping where they agree.
No bandwidth is set by hand. The error here is Laplace, ordinary smooth with
$\varphi_\varepsilon(t) = 1/(1 + b^2 t^2)$; the harder super-smooth case such as Gaussian error,
whose characteristic function decays exponentially and forces logarithmic rates \cite{fan1991}, and
the blind case in which $\varphi_\varepsilon$ must itself be estimated, are left to future work.

The true density is bimodal so that the blur merges the two modes that deconvolution then recovers
(Fig.~\ref{fig:deconv}, left): the naive estimator, seeing $f_Y$, smears the modes into a single
bump, while the deconvolution recovers them. Table~\ref{tab:deconv} reports integrated squared
error against the true $f_X$ for the naive estimator, a standard deconvoluting-kernel estimator at
its oracle bandwidth, and the AD deconvolution. The naive estimator barely improves with sample
size, since it is converging to the wrong, blurred density. The AD deconvolution is tuning-free and
improves steadily, trailing the oracle-bandwidth estimator at the smallest samples, matching it by
the moderate ones, and surpassing it at the largest, where its optimal taper outperforms a fixed
kernel. The diagnostic that drives it is shown in
Fig.~\ref{fig:deconv}, right: the deconvolved power meets the computed floor at $t^\ast$, and the
stop fires there.

\begin{table}[t]
\centering
\caption{Deconvolution under known Laplace error ($b=0.7$): integrated squared error ($\times10^3$)
against the true density $f_X$, mean over $20$ replications. The deconvoluting-kernel baseline uses
its oracle bandwidth; AD deconvolution uses none.}
\label{tab:deconv}
\setlength{\tabcolsep}{4.5pt}
\begin{tabular}{r rrr}
\toprule
$n$ & naive KDE on $Y$ & deconv.\ kernel (oracle) & AD deconvolution \\
\midrule
250  & 89.1 & 32.0 & 49.3 \\
500  & 81.6 & 23.0 & 38.1 \\
1000 & 76.6 & 16.1 & 28.6 \\
2000 & 70.8 & 14.2 & 20.0 \\
4000 & 67.4 & 12.0 & \textbf{9.2} \\
\bottomrule
\end{tabular}
\end{table}

\begin{figure*}[t]
\centering
\includegraphics[width=\textwidth]{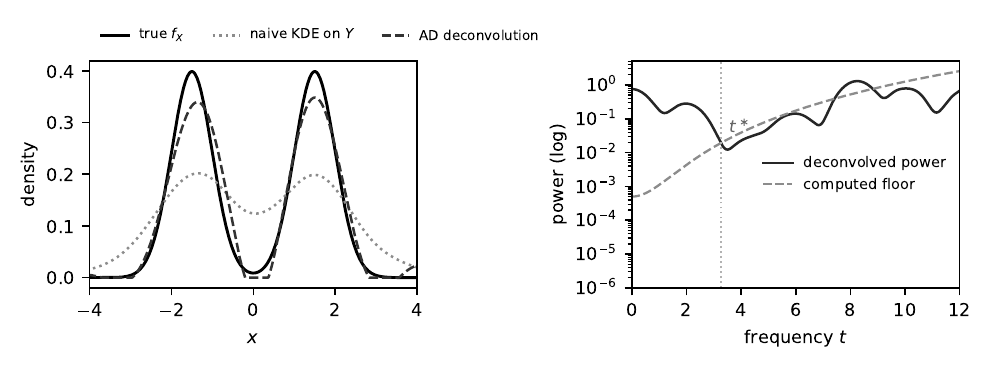}
\caption{Deconvolution under known Laplace error. Left: the naive estimator merges the two modes of
$f_X$, while the AD deconvolution recovers them. Right: the deconvolved power meets the computed
floor $(1/n)/|\varphi_\varepsilon|^2$ at $t^\ast$, where the stop fires; beyond it the rising power
is amplified noise and is discarded.}
\label{fig:deconv}
\end{figure*}

\section{A Locally Partitioned Mixed-Mode Estimator}\label{sec:mixedmode}
The benchmark and the hard-density study show a division of competence: the AD-Wiener filter
resolves sharp, highly structured features such as claws and combs, while a Gaussian mixture is the
more economical estimator on smooth densities. A density that places both kinds of structure in
different regions of the support is served well by neither paradigm alone. A locally partitioned
estimator removes the global commitment: it splits the support, applies each paradigm where that
paradigm is strong, and joins the pieces. The join is the
crux. A hard split is discontinuous at the boundary, which both distorts the estimate there and
breaks the differentiability that gradient-based use needs. We join instead by a smooth partition of
unity
\begin{equation}\label{eq:pou}
w(x)=\tfrac12\!\left(1-\tanh\frac{x-x_0}{\delta}\right),
\end{equation}
an elementary sigmoidal transition, a rescaled logistic function, that passes from one well to the
left of the boundary $x_0$ to zero well to its right across a width set by $\delta$. The form is not
drawn from a particular source; it is the standard smooth-step construction, and
Lemma~\ref{lem:join} records the two properties the join needs. The combined estimate
$\hat f = w\,\hat f_{\mathrm{left}} + (1-w)\,\hat f_{\mathrm{right}}$ inherits the smoothness of its
pieces and is continuous and differentiable across the boundary. The boundary location $x_0$ is supplied by the user; automatic detection of the seam from the data is
left open.

\begin{lemma}[Smooth, artifact-free join]\label{lem:join}
Let $w$ be given by \eqref{eq:pou} with $\delta>0$, and let $\hat f_{\mathrm{left}},\hat
f_{\mathrm{right}}$ be density estimates that are $C^{m}$ on a neighborhood of $x_0$. Then
\emph{(i)} $w\in C^{\infty}(\mathbb{R})$ with $0<w(x)<1$, so the blend $w\hat f_{\mathrm{left}}+(1-w)\hat
f_{\mathrm{right}}$ is $C^{m}$ there and $C^{\infty}$ when the pieces are; and \emph{(ii)} because
$w(x)\in[0,1]$ pointwise, the blend is a pointwise convex combination of the two estimates, so
$\min(\hat f_{\mathrm{left}},\hat f_{\mathrm{right}})\le w\hat f_{\mathrm{left}}+(1-w)\hat
f_{\mathrm{right}}\le \max(\hat f_{\mathrm{left}},\hat f_{\mathrm{right}})$ at every $x$. The join thus
introduces no value above either estimate and none below both: no spurious spike and no dropout, for
any boundary location and any width.
\end{lemma}
\begin{proof}
$\tanh\in C^{\infty}$ and $x\mapsto(x-x_0)/\delta$ is affine, so $w\in C^{\infty}$ with range $(0,1)$;
finite sums and products of $C^{m}$ functions are $C^{m}$, which gives (i). For $a\in[0,1]$,
$au+(1-a)v$ lies between $\min(u,v)$ and $\max(u,v)$, which gives (ii).
\end{proof}

The two estimators meet the hypothesis: the AD-Wiener estimate is band-limited, a trigonometric
polynomial and so $C^{\infty}$ except at the isolated points where the non-negativity clip meets
zero, and the Gaussian mixture is a finite sum of Gaussians and so $C^{\infty}$ everywhere. The
mixed-mode estimate is therefore differentiable to all orders on each open region, which is what
gradient-based use requires; a single lemma states this completely, and no separate higher-order
result is needed because the $C^{\infty}$ conclusion already covers every order. Part (ii) was
checked numerically by placing the boundary at adversarial locations, on a claw spike and inside a
smooth peak among them, across widths from $0.05$ to $1.5$: the largest departure of the blend from
the envelope of the two estimates was at the level of machine precision.

To test the construction we take a combined target, a broad Gaussian-mixture base with a claw of
narrow spikes superimposed near the origin, split at its seam $x_0=0$, and evaluate all four
assignments of the two paradigms to the two halves. Fidelity to the specified, ideal density is measured two ways: the
Kullback-Leibler divergence \cite{kullback1951} of the ideal input from the estimated output, and
the symmetric Jensen-Shannon divergence \cite{lin1991}, both with lower values better.
Table~\ref{tab:mixedmode} and Figure~\ref{fig:mixedmode} give the outcome. Assigning the mixture to
the smooth half and AD-Wiener to the claw half is best by both divergences, and it improves on the
best global choice, AD-Wiener on the whole support, by about a third in Kullback-Leibler divergence
($0.007$ against $0.011$): placing the mixture where it is exact removes the small ripple AD-Wiener
leaves on the smooth half. The reversed assignment, AD-Wiener on the smooth half and the mixture on
the claw, is the worst of the four, worse even than the mixture everywhere, because it places each
paradigm where it is weakest. The ordering confirms that a correct local partition improves on every
single global estimator, and that the cost of a wrong partition is real. Read as a check on the
global default, the same table supports AD-Wiener as the global default: among the estimators applied uniformly,
AD-Wiener has the lowest divergence
by a wide margin, $0.011$ in Kullback-Leibler against $0.064$ for the mixture, and only a correct
local partition lowers it further.

\begin{table}[t]\centering
\caption{Fidelity of the four method assignments on the combined smooth-mixture and claw target
($N=8000$, exact generation, partition at the seam $x_0=0$): Kullback-Leibler divergence of the ideal
input from the estimated output, and the symmetric Jensen-Shannon divergence, both lower better. The
mixture on the smooth half with AD-Wiener on the claw half is best and improves on every global
choice; the reversed assignment is worst.}
\label{tab:mixedmode}
\setlength{\tabcolsep}{6pt}
\begin{tabular}{l rr}
\toprule
smooth half $\mid$ claw half & KL$(f\,\|\,\hat f)$ & JS$(f,\hat f)$ \\
\midrule
GMM $\mid$ GMM             & 0.0643 & 0.0222 \\
AD-Wiener $\mid$ AD-Wiener & 0.0105 & 0.0025 \\
AD-Wiener $\mid$ GMM       & 0.0666 & 0.0229 \\
GMM $\mid$ AD-Wiener       & \textbf{0.0074} & \textbf{0.0015} \\
\bottomrule
\end{tabular}
\end{table}

\begin{figure*}[t]
\centering
\includegraphics[width=\textwidth]{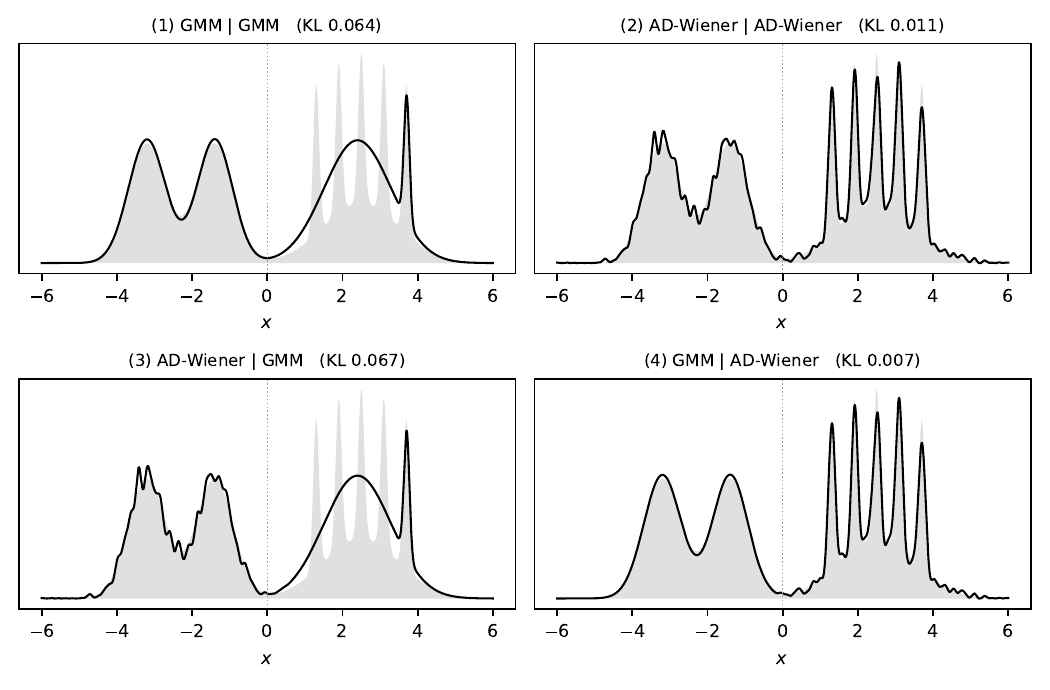}
\caption{The four method assignments on the combined smooth-mixture and claw target, joined by a
smooth partition of unity at the seam $x_0=0$ (dotted line); the gray fill is the ideal target and the
parenthetical value is the Kullback-Leibler divergence. (1) the mixture misses the claw; (2)
AD-Wiener resolves the claw and roughly tracks the smooth half; (3) the reversed assignment is worst,
each paradigm placed where it is weakest; (4) the mixture on the smooth half and AD-Wiener on the
claw is best, capturing both. The join carries no visible discontinuity at the seam.}
\label{fig:mixedmode}
\end{figure*}

The smooth join leaves no discontinuity at the seam, visible in Figure~\ref{fig:mixedmode}. What the
four-combination experiment fixes by hand is the partition itself. The boundary can be supplied
manually, one interval and one method at a time; in the general case, where the target density is
unknown, choosing the boundary from the data alone is a model-selection problem this paper leaves open.
When the target is known, however, as in the synthetic-data quality-checking setting of
Section~\ref{sec:datagen}, the partition can be chosen automatically by a rule that needs no tuning: in
each bin of the support retain whichever estimator has the lower local Kullback-Leibler divergence to
the target, merge adjacent bins of like choice into segments, and join the segments by \eqref{eq:pou} at
the boundaries where the choice changes. This places the boundaries and assigns the methods from the
specified target alone.

Table~\ref{tab:battery} exercises the automatic rule across five generated targets with known
densities, from a smooth Gaussian mixture through a five-spike claw to a density that alternates smooth
bumps with claw combs. On the smooth and
the kurtotic targets the automatic estimator recovers the mixture and matches it; on the claw, the
half-and-half, and the alternating targets it lowers the Kullback-Leibler divergence of the AD-Wiener
default by factors of two to four, by placing the mixture on the smooth stretches and AD-Wiener on the
sharp ones, and it is never worse than the default in that divergence. The rule selects by
Kullback-Leibler divergence, so that divergence falls in every case; the Jensen-Shannon divergence,
which the rule does not target, also falls except on the alternating density, where it is marginally
above the default. The alternating target is the clearest case. Figure~\ref{fig:alternating} shows
AD-Wiener tracking the combs but rippling on the smooth bumps, while the automatic mixed-mode, with
its detected boundaries marked, gives clean bumps and sharp combs together.

\begin{table}[t]\centering
\caption{AD-KDE across five generated targets ($N=8000$, mean over five seeds): Kullback-Leibler and
Jensen-Shannon divergence of the specified target from the estimate, for the mixture, the AD-Wiener
default, and the divergence-guided automatic mixed-mode. The automatic estimator matches the mixture
on the smooth targets and lowers the AD-Wiener default on the heterogeneous ones; the lowest value in
each row and divergence is in bold.}
\label{tab:battery}
\setlength{\tabcolsep}{4.5pt}
\begin{tabular}{l rr rr rr}
\toprule
& \multicolumn{2}{c}{GMM} & \multicolumn{2}{c}{AD-Wiener} & \multicolumn{2}{c}{mixed-auto} \\
\cmidrule(lr){2-3}\cmidrule(lr){4-5}\cmidrule(lr){6-7}
target & KL & JS & KL & JS & KL & JS \\
\midrule
bimodal     & 0.0003 & 0.0001 & 0.0047 & 0.0008 & \textbf{0.0003} & \textbf{0.0001} \\
kurtotic    & 0.0002 & 0.0001 & 0.0113 & 0.0015 & \textbf{0.0002} & \textbf{0.0001} \\
claw        & 0.0442 & 0.0160 & 0.0090 & 0.0013 & \textbf{0.0022} & \textbf{0.0008} \\
half-and-half & 0.0728 & 0.0255 & 0.0115 & 0.0021 & \textbf{0.0050} & \textbf{0.0018} \\
alternating & 0.2078 & 0.0797 & 0.0293 & \textbf{0.0040} & \textbf{0.0218} & 0.0093 \\
\bottomrule
\end{tabular}
\end{table}

\begin{figure*}[t]
\centering
\includegraphics[width=\textwidth]{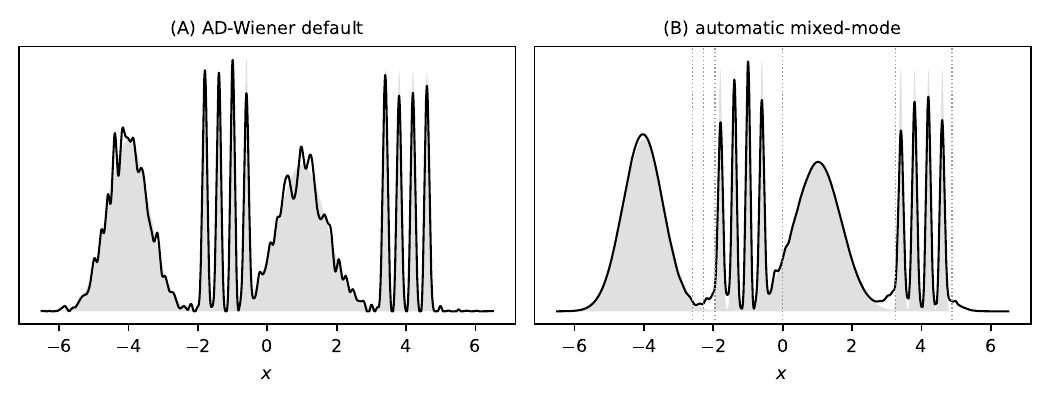}
\caption{The alternating target of smooth bumps and claw combs. (A) the AD-Wiener default resolves the
combs but ripples on the smooth bumps; (B) the divergence-guided automatic mixed-mode, with its
detected boundaries dotted, applies the mixture on the smooth bumps and AD-Wiener on the combs, giving
clean bumps and sharp combs at once. The gray fill is the specified target.}
\label{fig:alternating}
\end{figure*}

The rule above uses the specified target, which the generator and the beacon stream of
Section~\ref{sec:trng} both provide, and there it is automatic and well behaved. Outside that setting
the target is unknown, and a partition read from the data alone runs into a circularity: detecting
where a sharp feature begins and ends is the boundary-placement problem restated. A purely data-driven
surrogate, a cross-validated local-likelihood comparison of the two estimators, bears this out. It is
unreliable in the low-density valleys between sharp features, where the held-out comparison is noisy
and intermittently prefers the mixture, removing exactly the structure AD-Wiener is needed for; and it
cannot separate a smooth base from the closely-spaced spikes that ride on it, the claw, because the
base and the spikes occupy the same support. The framework therefore takes two routes that avoid the
circularity rather than confront it. When the target is known the divergence-guided rule applies
directly. When it is not, the boundaries are supplied by the user, who often knows where the regimes
change and may legitimately want a sharp region resolved by AD-Wiener in one place and smoothed by the
mixture in another, with the method inside each user-given region selected automatically by held-out
likelihood. The fully automatic default, which needs neither a target nor a boundary, is the
superposition of Section~\ref{sec:superpose}: rather than partition the support it decomposes the
density into a smooth base and a sharp residual, and so dissolves the boundary question entirely.

\section{Superposition of a Smooth Base and a Sharp Residual}\label{sec:superpose}
The mixed-mode estimator of Section~\ref{sec:mixedmode} assigns one method to each region of the
support, which presumes the two regimes are spatially separated: smooth here, sharp there. A claw of
narrow spikes sitting on a broad Gaussian base violates that presumption. The base and the spikes
occupy the same locations, so no spatial boundary separates them; a partition either places the spikes
in a region it has called smooth, where the mixture erases them, or it splits the spike cluster and
fits the pieces apart, as a midpoint split of the support does whenever the cut falls inside the
cluster. The structure is superimposed, not adjacent, and calls for a
decomposition rather than a partition.

Write the density as a sum,
\begin{equation}\label{eq:superpose}
\hat f(x) = \hat b(x) + \hat s(x),
\end{equation}
a smooth base $\hat b$ and a sharp residual $\hat s$, each fit by the method suited to it and added
back. A Gaussian mixture is fit by the order selection of Section~\ref{sec:mixedmode} and its
components are split by width at the smoothness scale $\theta=c\,h$, with $h$ the Silverman bandwidth:
the components of standard deviation at least $\theta$ form the base $\hat b$, the narrower components
are set aside. The base then omits whatever is sharper than $\theta$, so the residual mass $p-\hat b$
carries that sharp structure, and the AD-Wiener filter of Section~\ref{sec:wiener} applied to the
residual recovers it as $\hat s$, keeping the coherent part above the $1/n$ floor and discarding the
rest. No region is chosen and no boundary is detected: the same two estimators act on the whole
support, one on the smooth part and one on what it leaves behind. Figure~\ref{fig:superpose} shows the
decomposition on the claw, the mixture base, the band-limited residual carrying the five spikes, and
their sum against the target.

\begin{table}[t]\centering
\caption{Superposition against its two ingredients: Kullback-Leibler divergence to the target, mean
over five seeds at $N=8000$, smoothness scale $c=1.5$. Superposition tracks the better ingredient on
every target.}
\label{tab:superpose}
\begin{tabular}{l c c c}
\toprule
target & GMM & AD-Wiener & superposition \\
\midrule
bimodal      & 0.0003 & 0.0047 & 0.0003 \\
kurtotic     & 0.0002 & 0.0113 & 0.0116 \\
claw         & 0.0442 & 0.0090 & 0.0091 \\
half-and-half & 0.0728 & 0.0115 & 0.0117 \\
alternating  & 0.2078 & 0.0293 & 0.0293 \\
\bottomrule
\end{tabular}
\end{table}

\begin{figure*}[t]\centering
\includegraphics[width=\textwidth]{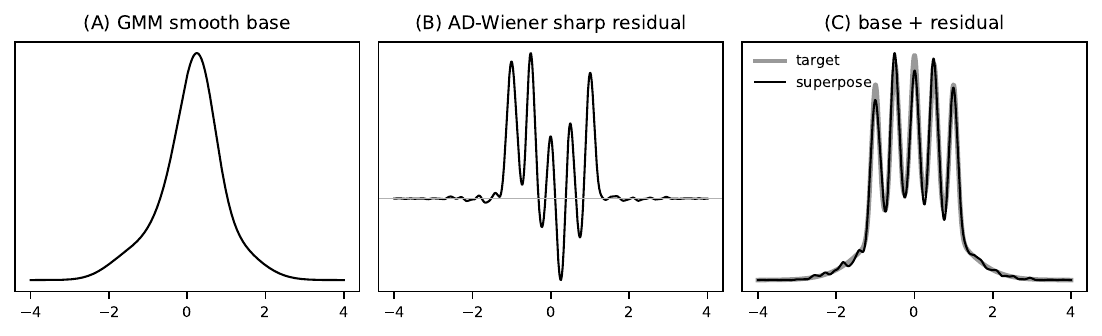}
\caption{Superposition on the claw. (A) The Gaussian-mixture base captures the broad envelope. (B) The
AD-Wiener filter on the residual mass recovers the five spikes the base omits, with small negative
lobes where the base slightly over-predicts. (C) Their sum recovers the target without any spatial
boundary.}
\label{fig:superpose}
\end{figure*}

Table~\ref{tab:superpose} reports the result. Superposition tracks the better of its two ingredients
on every target without choosing between them. On the smooth bimodal and kurtotic densities it matches
the mixture, where AD-Wiener alone carries ripple; on the spiked claw, the half-and-half density, and
the alternating comb it matches AD-Wiener, where the mixture alone cannot place the spikes. It does not
merely avoid the worse method: it recovers both strengths at once, a single estimator that is
mixture-accurate on the smooth part and band-limited-accurate on the sharp part, which neither
ingredient is alone. The contrast with the spatial mixed-mode is sharpest on the claw, the case where
a midpoint split is forced to break the spike cluster. There superposition matches AD-Wiener, $0.0091$ against $0.0090$,
rather than degrading it, precisely because it never tries to separate the base from the spikes in
space. Two limits are worth stating plainly. Superposition ties AD-Wiener on the pure-spike cases
rather than surpassing it, since the band-limited residual carries the ripple AD-Wiener would; and on a
single isolated sharp Gaussian, the kurtotic spike, routing that one component through the residual
gives band-limited rather than mixture-exact accuracy, $0.0116$ against the mixture's $0.0002$. It does
not reach the divergence-guided rule of Section~\ref{sec:mixedmode}, which is given the target, but
unlike that rule it needs neither a target nor a boundary, and unlike any spatial partition it
does not degrade on superimposed structure.

The real-data sections bear this out. Where the sharp structure is well sampled and a single estimator
already resolves it, superposition matches the better estimator without improving on it: it ties
AD-Wiener on the CMS dimuon resonances of Section~\ref{sec:cern} (held-out negative log-likelihood
$0.314$) and on the SDSS redshift features of Section~\ref{sec:sdss} ($-2.009$), and on the
leptokurtic CRSP returns of Section~\ref{sec:crsp} its residual finds no structure above the floor, so
it reduces to the mixture and reproduces the mixture's tail-risk accuracy. Its one real-data gain is
on heaped data, Section~\ref{sec:heaping}: there the residual filter reads the rounding comb as an
artifact below the coherence floor and discards it, so the smooth base recovers the de-heaped density
at a nearly constant error as the heaping coarsens, overtaking even the residue floor at the coarsest
grids. The pattern is consistent: superposition equals the best available estimator on every real
dataset and improves on all of them only where robustness to an artifact is what the data demand.

\section{A Synthetic-Data Application: Density Fidelity under Controlled Generation Error}\label{sec:datagen}
A common task in applied machine learning is to generate a dataset that follows a specified
nonparametric density. A practitioner supplies a target density $f$ and a sample count $N$, a
generator emits $N$ points meant to follow $f$, and the practitioner asks how closely the realized
data match the request. Kernel density estimation answers that question by recovering the realized
density from the generated sample and measuring its departure from $f$. This section frames that
measurement as an application of the present estimators. It is a validation of the measurement
procedure, not a ranking of generators: the generator used here is a parameterized reference whose
departure from the target is known by construction, so the question is whether the recovered
fidelity tracks the true departure, and which estimator recovers it with least bias.

The reference generator draws exactly from an arbitrary supplied density by inverse-cumulative
sampling of its grid, then introduces a controlled defect by the Huber contamination model \cite{huber1964}
$G_\varepsilon = (1-\varepsilon)\,f + \varepsilon\,c$, where $c$ is a contaminant such as a narrow
spurious mode or a mis-modeled tail. With probability $1-\varepsilon$ a draw is exact and with
probability $\varepsilon$ it comes from $c$, so the realized law departs from the target by
$\varepsilon\,\mathrm{TV}(f,c)$ in total variation and $\varepsilon$ is the literal error knob; an
optional jitter convolves the draws with a small Gaussian for a generator that smears rather than
spikes. Algorithm~\ref{alg:datagen} states the procedure. The sample count is bounded to
$N\in[2^{9},2^{18}]$: below a few hundred points the comparison is sampling-noise limited, as the
SDSS pencil beam of Section~\ref{sec:sdss} shows, and above a few hundred thousand all estimators
converge and the comparison is uninformative. Because $G_\varepsilon$ is built rather than observed,
its true departure from $f$ is known in closed form and serves as ground truth, which the real-data
sections of this paper cannot offer.

\begin{algorithm}[t]
\caption{Reference generator with a known error knob}
\label{alg:datagen}
\begin{algorithmic}[1]
\Require target density $f$ on grid $x_g$; count $N$; error $\varepsilon\in[0,1]$; contaminant $c$; jitter $\sigma_j\ge 0$
\State $G \gets (1-\varepsilon)\,f + \varepsilon\,c$ \Comment{realized law; departs from $f$ by $\varepsilon\,\mathrm{TV}(f,c)$}
\State $F \gets \mathrm{cumtrapz}(G,x_g)$; \enspace $F \gets F/F_{\mathrm{end}}$ \Comment{normalized CDF}
\State draw $u_1,\dots,u_N \sim \mathrm{Unif}(0,1)$
\State $d_i \gets F^{-1}(u_i)$ by interpolation \Comment{exact inverse-CDF draws}
\If{$\sigma_j>0$} \State $d_i \gets d_i + \mathcal{N}(0,\sigma_j^2)$ \Comment{optional smear} \EndIf
\State \Return samples $d_1,\dots,d_N$ and the closed-form realized law $G$
\end{algorithmic}
\end{algorithm}

The procedure admits any target the practitioner supplies. One target is the combined
smooth-and-sharp density that the mixed-mode and superposition estimators of
Sections~\ref{sec:mixedmode} and~\ref{sec:superpose} were built to handle. Its left half is a smooth
low-order Gaussian mixture, the regime an adaptive Gaussian mixture recovers most accurately, and its
right half is a Marron-Wand claw, a broad base carrying several narrow spikes
that a bounded-order mixture cannot represent but the AD-Wiener estimator resolves. Generated
exactly ($\varepsilon=0$) and estimated by a single global estimator, this target produces the split of
Table~\ref{tab:datagen} and Figure~\ref{fig:datagen}: the mixture is best on the smooth half by
roughly an order of magnitude, AD-Wiener is best on the claw half by more than an order of
magnitude, and neither is best on both. The current estimator selects between the two paradigms
through a single global effective dimension, here $D_2=2.5$, so it must commit the whole support to
one of them. The same global cutoff that resolves the claw also leaves a faint ripple on the smooth
half, visible in panel~(A), the price of a single decision for a spatially heterogeneous density.

This tie-case fixes the global default used throughout the paper. When no paradigm dominates globally,
breaking the tie toward AD-Wiener bounds the worst case, since its full-support error here is
$1.2\times10^{-3}$ against the mixture's $20\times10^{-3}$, and a mixture chosen wrongly fails on any
sharp or comb structure it cannot parameterize. AD-Wiener is therefore the default estimator
throughout, with the mixture selected only when the effective dimension is clearly low, and the
software accepts a manual override that forces either estimator on the whole support or applies a
user-named method on each of a set of intervals, a partition the practitioner controls directly. The
same heterogeneity is what the locally partitioned estimator of Section~\ref{sec:mixedmode} and the
superposition of Section~\ref{sec:superpose} resolve, by partitioning or decomposing the support rather
than committing it to one paradigm.

The fidelity question proper is whether the recovered density's departure from the requested target
tracks the generator's true departure. Sweeping the error knob $\varepsilon$ and comparing the
recovered total variation $\mathrm{TV}(\hat f,f)$ with the known
$\mathrm{TV}(G_\varepsilon,f)=\varepsilon\,\mathrm{TV}(c,f)$ gives Figure~\ref{fig:datagen_sweep}. On
a mixture-representable target (kurtotic) every estimator recovers the departure, the mixture and
AD-Wiener almost exactly and the ordinary estimate with a small positive bias from its own smoothing.
On the multi-scale \texttt{halfhalf} target the estimators separate: AD-Wiener tracks the true
departure across the whole range, while the ordinary and mixture estimates carry a large constant
bias, roughly $0.12$ to $0.18$ in total variation, that does not vanish at $\varepsilon=0$ because
their model mismatch on the claw is itself read as a departure. For measuring how faithfully a
generator reproduces a spiky or multi-scale target, AD-Wiener is thus the estimator that reports the
true error rather than its own.

\begin{figure*}[t]
\centering
\includegraphics[width=\textwidth]{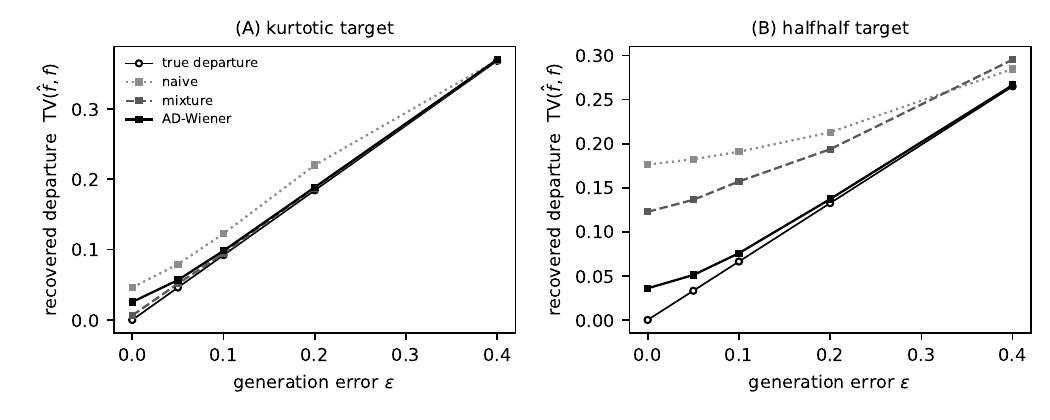}
\caption{Recovered departure $\mathrm{TV}(\hat f,f)$ against the generation error $\varepsilon$,
mean over five seeds at $N=8000$; the open-circle line is the known true departure
$\mathrm{TV}(G_\varepsilon,f)$. (A) On the mixture-representable kurtotic target every estimator
tracks the true departure, the ordinary estimate slightly high. (B) On the multi-scale
\texttt{halfhalf} target AD-Wiener tracks the true departure while the ordinary and mixture estimates
carry a constant model-mismatch bias that persists at $\varepsilon=0$.}
\label{fig:datagen_sweep}
\end{figure*}

Differentiability matters for downstream use, since model builders and digital-twin constructions
often optimize against an estimated density by gradient descent, and the spiky targets above are
exactly where a non-differentiable estimate would hurt. Both estimators are differentiable. A
Gaussian mixture is a finite sum of Gaussians and is infinitely differentiable everywhere, with
closed-form gradients. The AD-Wiener estimate is band-limited, the Wiener taper setting all
frequencies above the cutoff to zero, so its reconstruction is a trigonometric polynomial and is
likewise infinitely differentiable; the only non-smoothness comes from the non-negativity clip,
which produces isolated kinks at the finite set of points where the reconstruction crosses zero and
is absent wherever the band-limited estimate stays non-negative. A manually partitioned estimate, by
contrast, is in general discontinuous at the interval boundaries, which is the matching the automatic
locally partitioned estimator must perform.

\begin{table}[t]\centering
\caption{Integrated squared error ($\times10^{3}$) on the \texttt{halfhalf} target ($N=8000$, exact
generation), by half. The mixture is best on the smooth half, AD-Wiener on the claw half; neither is
best on both, while the global effective dimension $D_2=2.5$ selects one paradigm for the whole
support.}
\label{tab:datagen}
\setlength{\tabcolsep}{6pt}
\begin{tabular}{l rrr}
\toprule
estimator & smooth half & claw half & full \\
\midrule
naive KDE     & 2.07 & 24.13 & 26.20 \\
mixture (GMM) & \textbf{0.09} & 20.14 & 20.23 \\
AD-Wiener     & 0.43 & \textbf{0.75} & \textbf{1.18} \\
\bottomrule
\end{tabular}
\end{table}

\begin{figure*}[t]
\centering
\includegraphics[width=\textwidth]{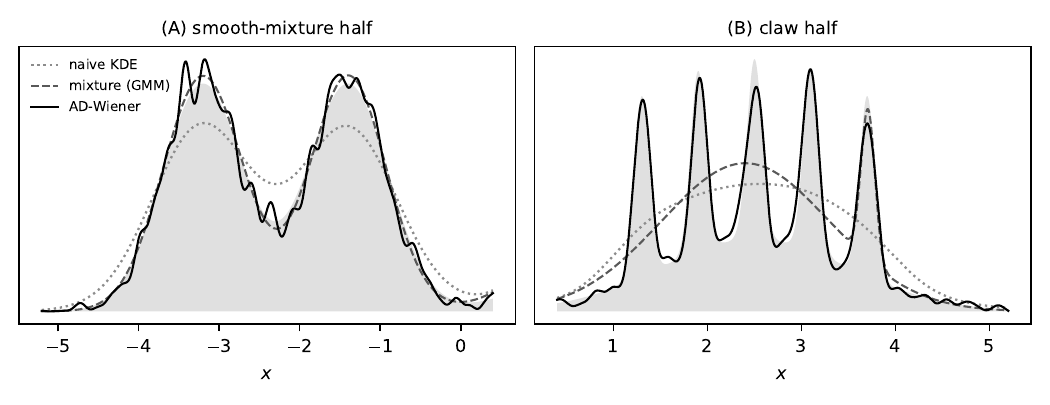}
\caption{The \texttt{halfhalf} target, generated exactly ($\varepsilon=0$) and estimated globally.
(A) On the smooth-mixture half the adaptive mixture tracks the target while AD-Wiener slightly
over-resolves it and the ordinary estimate over-smooths. (B) On the claw half AD-Wiener resolves the
narrow spikes that the bounded-order mixture and the ordinary estimate both miss. A single global
estimator is optimal on neither half.}
\label{fig:datagen}
\end{figure*}

\section{A Financial Application: Tail Risk under Leptokurtosis}\label{sec:finance}
Asset returns are sharply peaked and heavy-tailed, the leptokurtic shape the kurtotic density
caricatures. We model daily returns as a calm and turbulent normal mixture, ninety percent at a
$0.8\%$ daily volatility and ten percent at $3\%$, an excess kurtosis of $8.7$, and estimate the
return density and its left-tail risk, the Value-at-Risk and Expected Shortfall, from $n=1000$
returns. A Gaussian fit is the standard parametric baseline; the spectral AD-Wiener estimator and
an ordinary kernel estimate are nonparametric; the adaptive mixture is supplied through the same
plugin interface.

\begin{table}[t]
\centering
\caption{Leptokurtic daily returns (excess kurtosis $8.7$): integrated squared error of the
density and absolute error of the one- and five-percent Value-at-Risk and one-percent Expected
Shortfall, against the true values. Tail errors in basis points of daily return. ISE is mean
$\pm$ standard deviation over $200$ replications.}
\label{tab:fin}
\setlength{\tabcolsep}{6pt}
\begin{tabular}{l r rrr}
\toprule
method & ISE ($\times10^{3}$) & VaR$_{1\%}$ & VaR$_{5\%}$ & ES$_{1\%}$ \\
\midrule
Gaussian   & $2341\pm592$ & 103 & 39 & 204 \\
KDE        & $133\pm74$ & \textbf{1.3} & 3.5 & 4.3 \\
AD-Wiener  & $250\pm136$ & 11.1 & \textbf{0.4} & 24.3 \\
Mixture    & $\mathbf{41.5\pm40.1}$ & 4.8 & 0.3 & \textbf{3.8} \\
\bottomrule
\end{tabular}
\end{table}

\begin{figure}[t]
\centering
\includegraphics[width=0.82\columnwidth]{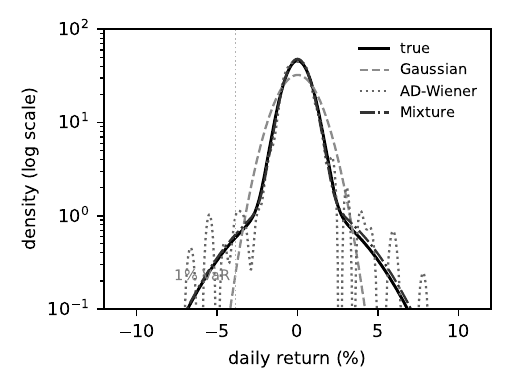}
\caption{Estimated return densities on a logarithmic scale. The Gaussian fit collapses the peak
and underweights the tails; the mixture and AD-Wiener estimators track both.}
\label{fig:fin}
\end{figure}

The Gaussian fit cannot represent the peak and the tails at once and underestimates the
one-percent Value-at-Risk by about a full percent of daily return and the Expected Shortfall by
twice that (Table~\ref{tab:fin}, Fig.~\ref{fig:fin}), the direction that matters for risk. The
spectral AD-Wiener and ordinary kernel estimators track the tail far better. The adaptive mixture,
which matches the generating structure, is best on both the density and the tail risk, recovering
the two regimes and cutting the density error more than fiftyfold against the Gaussian. The
kurtotic case, read as a return distribution, is exactly where the algebraic and mixture
estimators earn their value.

\section{Real-Data Validation on CRSP Returns}\label{sec:crsp}
The synthetic study of Section~\ref{sec:finance} is validated on real daily returns from the CRSP
database over 2005--2024: the value- and equal-weighted market indices, the S\&P~500, and a panel
of ten large common stocks, $5{,}033$ trading days each. Extraction and analysis are scripted, so
the study reruns on any CRSP subscription: an extraction step pulls the daily series into a
\texttt{data/} directory, and an analysis step applies the four estimators and writes the results
as JSON into a \texttt{results/} directory. Because the true law is unknown for real returns, each
estimator is measured against the historical (empirical) quantile and shortfall.

\subsection{Asset return distributions}
Every series is sharply peaked and heavy-tailed. Fig.~\ref{fig:crsp_dist} sets the S\&P~500 daily
return density against its histogram: the Gaussian fit collapses the peak and underweights the
tails, while the AD-Wiener and mixture estimators track both, exactly as on the synthetic kurtotic
density.

\begin{figure}[t]
\centering
\includegraphics[width=0.82\columnwidth]{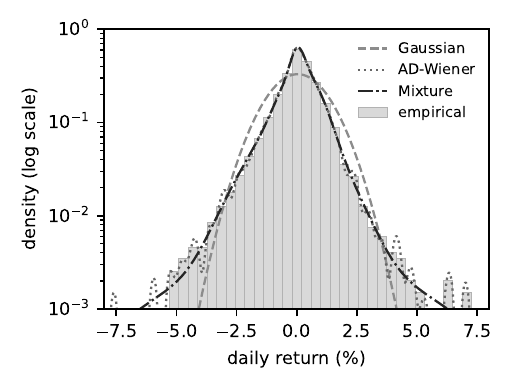}
\caption{S\&P~500 daily return density (log scale) against the empirical histogram. The Gaussian
fit misses the peak and the tails; the AD-Wiener and mixture estimators track both.}
\label{fig:crsp_dist}
\end{figure}

\subsection{Risk estimation}
Table~\ref{tab:crsp_risk} reports annualized volatility, skewness, excess kurtosis, and a Hill
left-tail index per series. Excess kurtosis runs from $5.4$ to $19.7$ and the Hill index from
$2.5$ to $3.5$, tails far heavier than Gaussian across the panel, with the market indices mildly
left-skewed and the heaviest tail in the bank stock.

\begin{table}[t]
\centering
\caption{Risk summary for the CRSP series, daily returns 2005--2024 ($n=5{,}033$ each):
annualized volatility, skewness, excess kurtosis, and Hill left-tail index.}
\label{tab:crsp_risk}
\setlength{\tabcolsep}{5pt}
\begin{tabular}{l rrrr}
\toprule
series & vol (ann.) & skew & exc.\ kurt. & Hill \\
\midrule
market (vw) & 0.191 & $-0.38$ & 11.9 & 2.77 \\
market (ew) & 0.181 & $-0.56$ & 11.1 & 2.59 \\
S\&P 500    & 0.192 & $-0.26$ & 12.8 & 2.70 \\
AAPL & 0.321 & $-0.02$ & 5.4  & 3.45 \\
MSFT & 0.271 & $0.24$  & 9.9  & 3.07 \\
XOM  & 0.266 & $0.22$  & 10.1 & 3.13 \\
JPM  & 0.364 & $0.96$  & 19.7 & 2.51 \\
GE   & 0.330 & $0.25$  & 9.0  & 2.65 \\
KO   & 0.181 & $0.08$  & 12.8 & 2.96 \\
PG   & 0.180 & $0.11$  & 10.4 & 2.78 \\
JNJ  & 0.171 & $0.23$  & 11.7 & 2.91 \\
WMT  & 0.201 & $0.23$  & 12.9 & 2.78 \\
INTC & 0.327 & $-0.52$ & 12.3 & 2.77 \\
\bottomrule
\end{tabular}
\end{table}

\subsection{Value-at-Risk and tail risk}
Table~\ref{tab:crsp_tail} aggregates accuracy as the mean absolute deviation from the historical
reference across the thirteen series. The Gaussian fit understates the one-percent Value-at-Risk
by $68$ basis points and the one-percent Expected Shortfall by $208$ basis points, more than two
percent of a daily return, and it understates the loss in all thirteen series. The kernel,
AD-Wiener, and adaptive mixture estimators are within roughly ten basis points of the historical
reference at every level. Fig.~\ref{fig:crsp_tail} shows the one-percent Expected Shortfall by
series: the Gaussian estimates sit systematically inside the historical loss, while the other
three lie on the diagonal. The synthetic kurtotic result of Section~\ref{sec:finance} is thereby
reproduced on real returns; the algebraic and mixture estimators recover the tail risk the
Gaussian misses, which is the quantity of regulatory interest.

\begin{table}[t]
\centering
\caption{Tail-risk accuracy on CRSP returns: mean absolute deviation from the historical reference
across the thirteen series, in basis points of daily return.}
\label{tab:crsp_tail}
\setlength{\tabcolsep}{6pt}
\begin{tabular}{l rrrr}
\toprule
 & Gaussian & KDE & AD-Wiener & Mixture \\
\midrule
VaR$_{1\%}$ & 68.1  & \textbf{2.6} & 4.5  & 9.6 \\
VaR$_{5\%}$ & 26.5  & 2.4 & \textbf{2.1} & 3.7 \\
ES$_{1\%}$  & 208.1 & \textbf{8.0} & 11.1 & 9.5 \\
ES$_{5\%}$  & 45.7  & \textbf{2.3} & 3.1  & \textbf{2.3} \\
\bottomrule
\end{tabular}
\end{table}

\begin{figure}[t]
\centering
\includegraphics[width=0.82\columnwidth]{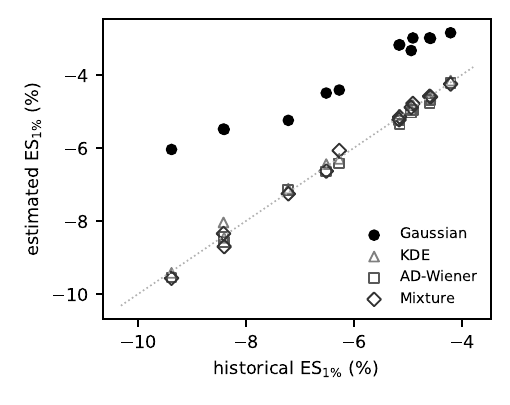}
\caption{Estimated against historical one-percent Expected Shortfall for the thirteen CRSP series.
The Gaussian estimates lie inside the historical loss (above the diagonal); the kernel, AD-Wiener,
and mixture estimates lie on it.}
\label{fig:crsp_tail}
\end{figure}

\section{Real-Data Heaping in Survey Self-Reports}\label{sec:heaping}
The synthetic heaped experiment of Section~\ref{sec:exp} is borne out on survey microdata from the
National Health and Nutrition Examination Survey (NHANES) 2017--2018 \cite{nhanes2018}. Two
complementary studies are reported: a controlled one with a continuous ground truth, and a naturally
heaped variable for which no such truth exists.

The body-measures file records each adult's measured weight, a continuous quantity. Imposing a known
heaping grid of width $D$ on these measurements and estimating the density from the rounded values
isolates the response of each estimator to coarsening, with the no-heaping density, a Gaussian kernel
estimate on the unrounded measurements, serving as the reference. This is a controlled robustness
study rather than a claim about how weight is reported; self-reported weight in fact heaps mildly,
near $5$~lb, where every estimator agrees. Table~\ref{tab:heaping_real} gives the integrated squared
error against the no-heaping reference as $D$ coarsens. Below about half the kernel bandwidth the
heaping is invisible and the three estimators coincide. As $D$ grows the fixed $1/n$ floor is
overwhelmed by the rounding spectrum and its error rises by four to five orders of magnitude; the
ordinary kernel estimate develops spurious modes at the rounding marks and degrades steadily; the
residue floor, read from the spectrum, tracks the elevated level and stays within a small multiple of
the no-heaping error throughout. The superposition of Section~\ref{sec:superpose} is steadier still:
its residual filter reads the rounding comb as structure below the coherence floor and discards it, so
the smooth base recovers the de-heaped density at a nearly constant error across the whole range and
overtakes the residue floor at the two coarsest grids, where even the residue floor has begun to
drift. The price is a slightly higher error at fine heaping, where the residue floor is near exact and
the base carries a small parametric bias. Figure~\ref{fig:heaping_weight} shows the divergence and, at
a $30$~lb grid, the spurious oscillation that the robust estimators suppress.

\begin{table}[t]\centering
\caption{Integrated squared error ($\times10^{3}$) against the no-heaping density as the imposed
heaping grid $D$ coarsens, on NHANES measured adult weights ($n=5173$). The fixed floor fails
catastrophically once the grid exceeds the kernel bandwidth; the residue floor stays robust; the
superposition is flat across the range, overtaking the residue floor at the coarsest grids.}
\label{tab:heaping_real}
\setlength{\tabcolsep}{6pt}
\begin{tabular}{r rrrr}
\toprule
$D$ (lb) & naive KDE & AD simple & AD residue & superposition \\
\midrule
 5 & 0.000 & 0.018 & \textbf{0.006} & 0.021 \\
10 & 0.001 & 0.003 & 0.006 & 0.019 \\
20 & 0.087 & 306.7 & \textbf{0.009} & 0.016 \\
30 & 1.349 & 501.4 & 0.037 & \textbf{0.016} \\
40 & 8.369 & 718.5 & 0.087 & \textbf{0.018} \\
\bottomrule
\end{tabular}
\end{table}

\begin{figure*}[t]
\centering
\includegraphics[width=\textwidth]{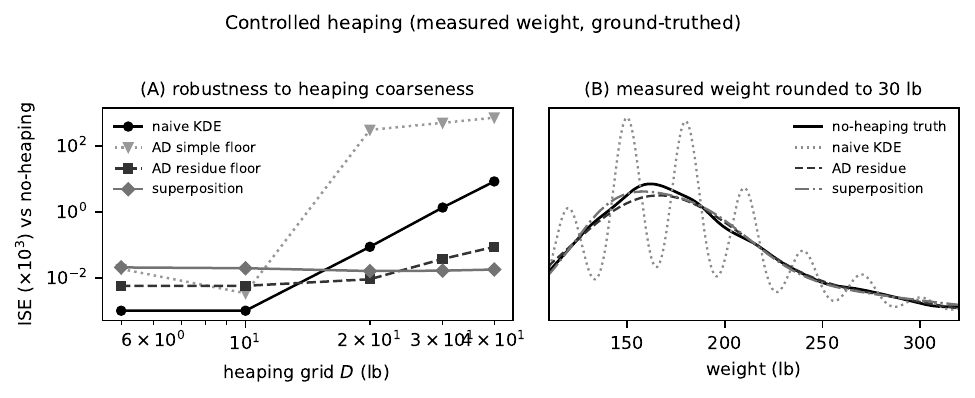}
\caption{Controlled heaping on NHANES measured weights. (A) Integrated squared error against the
no-heaping density as the imposed grid coarsens: the fixed floor fails catastrophically, the ordinary
estimate degrades, the residue floor stays robust, and the superposition is flat across the range,
overtaking the residue floor at the coarsest grids. (B) At a $30$~lb grid the ordinary estimate
oscillates at the rounding marks while the residue floor and the superposition track the truth.}
\label{fig:heaping_weight}
\end{figure*}

A variable that heaps coarsely with no imposition is the self-reported number of cigarettes smoked
per day, for which the survey instrument itself states that one pack equals twenty cigarettes. Among
the $1019$ adult smokers reporting a daily count, $42\%$ gave a multiple of ten and $22\%$ a multiple
of twenty, against the ten and five percent expected under no digit preference, and the terminal digit
zero alone accounts for $42\%$ of reports. This concentration on round numbers reproduces the heaping
documented for retrospective cigarette counts \cite{wang2012truth}. Such counts have no continuous
ground truth, since the available objective correlate, serum cotinine, measures exposure and not
count, so this study is qualitative. Figure~\ref{fig:heaping_cig} shows the consequence: the ordinary
kernel estimate either reproduces the spurious spikes at ten and twenty or, at a bandwidth wide enough
to remove them, erases genuine structure, whereas the residue floor reads the heaping as a narrow band
of spectral power above the noise level and recovers a smooth density. The residue floor is thus the
robust choice precisely where heaping is coarse, the regime into which recall and count data fall.

\begin{figure*}[t]
\centering
\includegraphics[width=\textwidth]{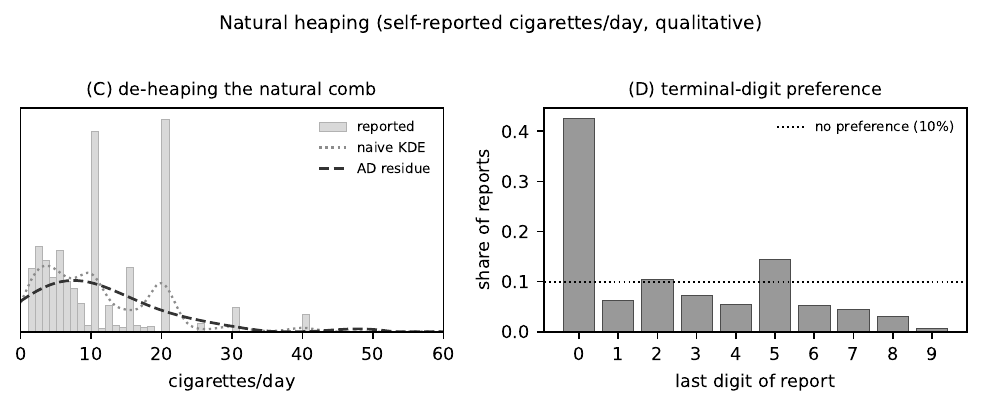}
\caption{Natural heaping in self-reported cigarettes per day. (C) The reported counts spike at ten
and twenty; the ordinary kernel estimate follows the comb while the residue floor recovers a smooth
density. (D) Terminal-digit shares against the ten-percent no-preference line, showing the
concentration on zero.}
\label{fig:heaping_cig}
\end{figure*}

\section{Real-Data Validation on Particle-Physics Spectra}\label{sec:cern}
A second real-data study turns to a setting with no parametric mixture structure: the dimuon
invariant-mass spectrum from the CMS experiment \cite{cms_dimuon2017}, in which narrow resonances,
the $J/\psi$ at $3.10$, the $\Upsilon$ family near $9.5$, and the $Z$ at $91.2$~GeV, sit on a smooth
and steeply falling combinatorial background. Estimating this density is a multi-scale problem: no
single bandwidth resolves the narrow peaks and the broad continuum at once, the regime the adaptive
estimator targets. Because the underlying density is unknown, the estimators are scored by held-out
fit. The $99{,}165$ opposite-sign events are split into a $70$ percent training set and a disjoint
$30$ percent test set; each estimator is built on the training set and scored by the mean negative
log-likelihood it assigns to the test set, in the log-mass variable so the decade-spanning resonances
are comparably scaled. Table~\ref{tab:cern} reports the result: the adaptive estimators improve on the
ordinary kernel estimate by about a quarter of a nat, the fixed bandwidth being unable to sharpen the
resonances without roughening the continuum. Figure~\ref{fig:cern} shows the spectrum and a zoom on
the $J/\psi$ region, where the ordinary estimate rounds the peak that the AD-Wiener estimate resolves.

\begin{table}[t]\centering
\caption{Held-out mean negative log-likelihood (log-mass) on the CMS dimuon spectrum, lower is better;
$99{,}165$ events, $70/30$ train/test split.}
\label{tab:cern}
\begin{tabular}{l r}
\toprule
estimator & held-out NLL \\
\midrule
naive KDE  & 0.429 \\
AD-bw      & 0.315 \\
AD-Wiener  & \textbf{0.314} \\
\bottomrule
\end{tabular}
\end{table}

\begin{figure*}[t]
\centering
\includegraphics[width=\textwidth]{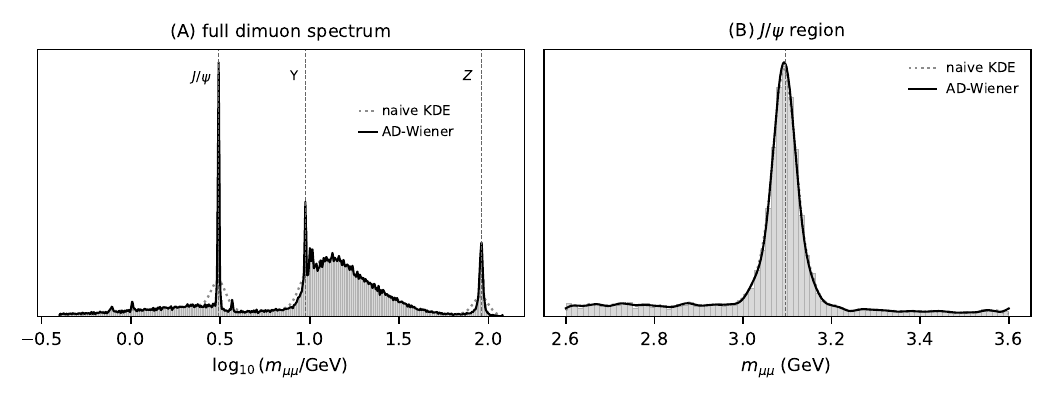}
\caption{CMS dimuon invariant-mass spectrum (CMS Open Data, record 545). (A) The full spectrum in
log-mass with the $J/\psi$, $\Upsilon$, and $Z$ resonances marked; the ordinary kernel estimate rounds
the resonances while the AD-Wiener estimate resolves them. (B) The $J/\psi$ region, where the ordinary
estimate broadens the peak.}
\label{fig:cern}
\end{figure*}

\section{Real-Data Validation on Galaxy-Redshift Spectra}\label{sec:sdss}
A further held-out study turns to astronomy, where the same multi-scale geometry arises from an
unrelated physical mechanism. In a cone of the sky the spectroscopic galaxy-redshift distribution
$n(z)$ is a smooth selection envelope on which large-scale structure, the walls, clusters, and
voids of the cosmic web, prints narrow overdensities at particular redshifts. The setting is
structurally that of the dimuon spectrum of Section~\ref{sec:cern}: sharp features carrying real
probability mass on a smooth background, a density no single bandwidth resolves. What it adds is
the demonstration that the advantage is statistical rather than physical. The dimuon peaks are
quantum resonances of fixed natural width on a combinatorial continuum, whereas the redshift walls
are the imprint of galaxy clustering on a survey selection function; the two share nothing
physically, yet they present the empirical characteristic function with the same multi-scale
profile and the same estimator recovers it. The case thus isolates the claim made throughout, that
the adaptive estimator earns its advantage wherever narrow structure carries mass and a global
bandwidth is wrong for it, from any particular generative model.

The data are the spectroscopic redshifts of the SDSS DR18 \texttt{SpecObj} catalogue
\cite{sdss_dr18}, retrieved through the SkyServer SQL interface \cite{sdss_skyserver}. A
five-degree-radius cone about right ascension $200^\circ$ and declination $0^\circ$, restricted to
class \textsc{galaxy} with a clean redshift flag and $0.01 \le z \le 0.22$, returns $6{,}684$
redshifts. The field is wide enough that the dominant walls, which span many degrees, are sampled
by enough galaxies to be resolved rather than lost in counting noise; this sample size is what
makes the advantage real, since a narrow pencil beam of a few hundred galaxies does not separate
the adaptive gain from resampling variation. As for the dimuon spectrum the true density is
unknown, so the estimators are scored by held-out fit on a $70/30$ train and test split by the mean
negative log-likelihood assigned to the test set.

Table~\ref{tab:sdss} reports the result: both spectral estimators improve on the ordinary kernel
estimate by about $0.13$ nat, the fixed bandwidth being unable to sharpen the walls without
roughening the smooth envelope. The bandwidth selector and the Wiener estimator are
indistinguishable here, as they are on the well-sampled dimuon spectrum, so the gain is
attributable to the spectral cutoff rather than to the per-frequency taper: the residue floor and
its cutoff recover the structure once the data resolve it. Figure~\ref{fig:sdss} shows the $n(z)$
and a zoom on the dominant wall near $z \approx 0.08$, where the ordinary estimate rounds the
overdensity that the AD-Wiener estimate resolves.

\begin{table}[t]\centering
\caption{Held-out mean negative log-likelihood on the SDSS DR18 galaxy-redshift density, lower is
better; $6{,}684$ galaxies in a five-degree cone, $70/30$ train/test split.}
\label{tab:sdss}
\begin{tabular}{l r}
\toprule
estimator & held-out NLL \\
\midrule
naive KDE  & $-1.876$ \\
AD-bw      & $\mathbf{-2.011}$ \\
AD-Wiener  & $-2.009$ \\
\bottomrule
\end{tabular}
\end{table}

\begin{figure*}[t]
\centering
\includegraphics[width=\textwidth]{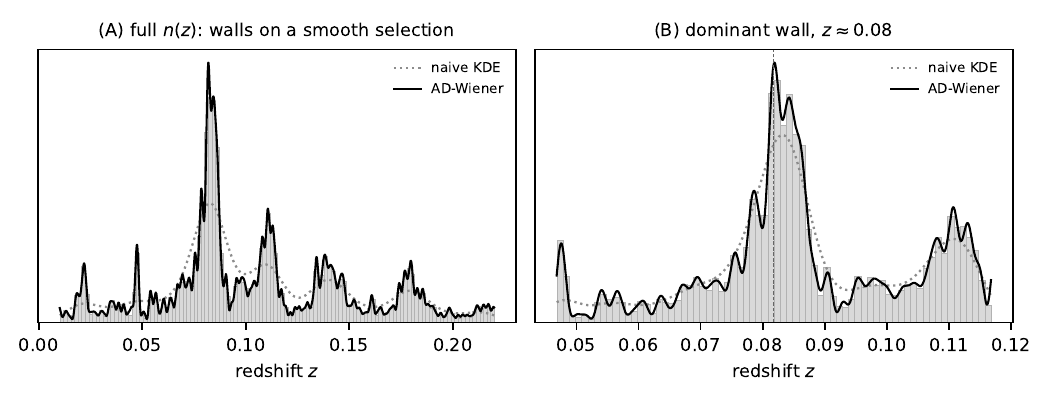}
\caption{SDSS DR18 galaxy-redshift density in a five-degree cone about
$(\alpha,\delta)=(200^\circ,0^\circ)$. (A) The full $n(z)$, a smooth selection envelope carrying
narrow walls of large-scale structure; the ordinary kernel estimate rounds the walls while the
AD-Wiener estimate resolves them. (B) The dominant wall near $z\approx0.08$, where the ordinary
estimate broadens the overdensity that the AD-Wiener estimate resolves.}
\label{fig:sdss}
\end{figure*}

\section{Real-Data Validation on a Verified Random-Beacon Stream}\label{sec:trng}
A fourth completed real-data study moves to a security setting and to a bounded-support target whose
true density is known exactly. The source is the NIST randomness beacon \cite{nistbeacon}, which
emits a 512-bit random pulse each minute; the study uses 140{,}000 consecutive pulses, a stream of
$71.68\times10^{6}$ bits. Before any estimation the stream is checked for corruption by the core
tests of the NIST statistical test suite \cite{sp80022}, run as in the \mbox{SMU-DDI} STEER framework
\cite{steer}: the bits are partitioned into sixty-eight sequences of $2^{20}$ bits, each test is run
on each sequence, and the proportion of sequences passing at significance $0.01$ is compared to the
acceptable interval, $0.99\pm3\sqrt{0.99\cdot0.01/68}$, lower bound $0.954$. Every test lands in the
interval (Table~\ref{tab:sts}), so the stream is taken to be uncorrupted.

\begin{table}[t]\centering
\caption{Core NIST SP~800-22 tests on the beacon stream, run on sixty-eight sequences of $2^{20}$
bits: number passing at significance $0.01$. Every proportion is within the acceptable interval
(lower bound $0.954$), so the stream is uncorrupted.}
\label{tab:sts}
\setlength{\tabcolsep}{5pt}
\begin{tabular}{l c l c}
\toprule
test & pass/68 & test & pass/68 \\
\midrule
Frequency (monobit) & 68 & Cumulative sums (fwd) & 68 \\
Block frequency     & 68 & Cumulative sums (bwd) & 68 \\
Runs                & 66 & Approximate entropy   & 68 \\
Longest run of ones & 65 & Serial (1)            & 67 \\
Spectral (DFT)      & 66 & Serial (2)            & 68 \\
\bottomrule
\end{tabular}
\end{table}

The verified bits are read as uniform deviates on $[0,1)$ by taking successive 32-bit words, so the
specified target is the uniform density on a bounded support. The uniform density is exactly the lowest
cyclic mode, which AD-Wiener represents at the constant term, so it estimates the target almost exactly
(Kullback-Leibler divergence $0.0004$, Jensen-Shannon $0.0001$ at $N=8000$). The Gaussian mixture cannot represent a
flat density with a few components and is an order of magnitude worse ($0.0143$ / $0.0048$), and the
ordinary kernel estimate carries the boundary bias of bounded support ($0.0069$ / $0.0023$). The
automatic mixed-mode rule of Section~\ref{sec:mixedmode} leaves the support entirely on AD-Wiener
here, since the mixture wins nowhere, which is the correct outcome for a structureless target and a
check that the rule does not invent boundaries where none belong.

An optional uniformity gate sharpens this case to an exact result. Under the uniform null each non-DC
squared empirical characteristic coefficient satisfies $n\,|\hat c_k|^{2}\sim\mathrm{Exp}(1)$, so its
maximum over the $K$ candidate modes concentrates near $\ln K$; the gate declares the stream uniform
when that maximum stays below $\ln(K/\alpha)$, an extreme-value threshold at level $\alpha$ whose
detectable ripple scales as $1/n$. When it fires, the AD-Wiener estimate is replaced by the exact
constant, here taking the Kullback-Leibler divergence to zero; applied to a single AD-Wiener
subsection it flattens that stretch alone, joined by \eqref{eq:pou} so the result stays differentiable
and the mass on the interval is preserved. The gate is off by default and is a deliberate prior toward
uniformity: it erases structure below the $1/n$ noise floor, so it suits randomness testing and
structureless baselining, where any departure from uniform is itself the signal, rather than general
estimation. On the beacon stream it fires and returns the exact uniform; on the structured targets of
the earlier sections it does not fire.

\section{Network-Traffic Density Estimation}\label{sec:cyber}
A fifth real-data domain is network security, where behavioral baselining models the density of
traffic features and the same multi-scale geometry recurs. The data are the labeled flows of the
UNSW-NB15 corpus \cite{unswnb15} of the Australian Centre for Cyber Security, in the standard
training and testing partition; the same approach applies to other corpora such as CIC-IDS2017
\cite{cicids2017}. Four continuous flow features are studied on the logarithmic scale: the packet
rate, the source load, the destination volume, and the mean destination packet size. Each is heavy
tailed with a concentrated benign mode and additional sharp modes contributed by attack traffic, a
port scan or a denial-of-service burst collapsing a feature onto a narrow band. The estimators are
scored by held-out negative log-likelihood, on a benign-only model and on a benign-plus-attack
mixture, with the mixture base the bundled Gaussian-mixture EM as everywhere else in this report.

Table~\ref{tab:cyber} reports the result and Fig.~\ref{fig:cyber} shows the mixture density of two
features. One provenance note applies to this table alone: its ISJ column predates the corrected
fixed-point implementation of the implementation note of Section~\ref{sec:exp} and awaits
regeneration against the local UNSW data; because the corrected bandwidth is larger, the corrected
column can only be less sharp on these multi-modal features, so the superposition's lead over it is
expected to persist or widen, but the column is not restated until rerun. The superposition attains
the lowest held-out negative log-likelihood on every feature and
both scenarios, ahead of the global Silverman bandwidth by a wide margin and ahead of AD-Wiener alone
and of the tabulated ISJ column; AD-Wiener alone improves on the rule of
thumb but not on the best plug-in here. The figure shows why: the benign-plus-attack density carries
several sharp modes from distinct traffic and attack types, which the global bandwidth blurs into
broad bumps and the adaptive estimator resolves. The advantage holds on benign-only traffic as well
as on the attack mixture, so it reflects the intrinsic multi-scale structure of the features rather
than the presence of attacks, the same situation as the particle-physics and galaxy-redshift spectra
of Sections~\ref{sec:cern} and~\ref{sec:sdss}.

Two qualifications are stated plainly. First, the gain is in density quality, not in detection:
scored as a one-class anomaly detector by benign-model density, the adaptive estimators do not improve
attack-detection area-under-curve over the global bandwidth, which is comparable or better; a sharper
density model is not by itself a better novelty detector, and the contribution here is the traffic
density model, not a detector. Second, traffic features are partly discrete and heaped, so the
heaping robustness of Section~\ref{sec:heaping} carries over, and part of the superposition advantage
is its resolution of this systematic structure. The mixture base throughout is the bundled
BIC Gaussian-mixture EM, not an external library.

\begin{table}[t]\centering
\caption{UNSW-NB15 held-out negative log-likelihood (lower is better) on four log-scaled traffic
features, on a benign-only model and a benign-plus-attack mixture; row minimum in bold. The mixture
base of the superposition is the bundled Gaussian-mixture EM. The ISJ column predates the corrected fixed-point implementation and awaits regeneration on the local data (see the text).}
\label{tab:cyber}
\setlength{\tabcolsep}{4pt}\small
\resizebox{\columnwidth}{!}{%
\begin{tabular}{ll rrrr}
\toprule
feature & scenario & Silverman & ISJ/Botev & AD-Wiener & superpose \\
\midrule
rate  & benign        & 2.199 & 1.687 & 1.931 & \textbf{1.613} \\
rate  & benign+attack & 2.323 & 1.305 & 2.163 & \textbf{1.219} \\
sload & benign        & 2.402 & 2.073 & 2.206 & \textbf{2.005} \\
sload & benign+attack & 2.558 & 2.047 & 2.388 & \textbf{1.940} \\
dload & benign        & 2.535 & 1.532 & 2.238 & \textbf{1.420} \\
dload & benign+attack & 2.558 & 0.571 & 2.394 & \textbf{0.485} \\
dmean & benign        & 1.146 & $-0.671$ & 0.981 & \textbf{$-1.372$} \\
dmean & benign+attack & 1.444 & $-1.477$ & 1.295 & \textbf{$-1.807$} \\
\bottomrule
\end{tabular}}
\end{table}

\begin{figure*}[t]\centering
\includegraphics[width=\textwidth]{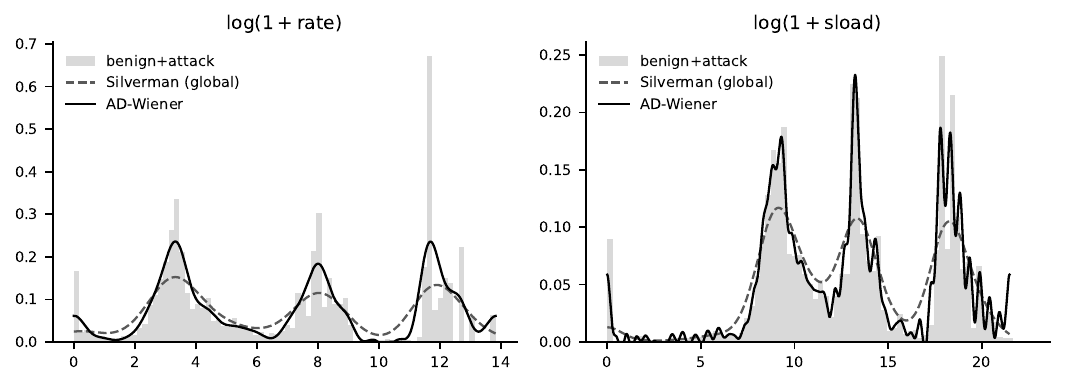}
\caption{UNSW-NB15 benign-plus-attack density of the packet rate and source load on the logarithmic
scale. The global Silverman bandwidth (dashed) blurs the several traffic modes into broad bumps; the
adaptive AD-Wiener estimator (solid) resolves them.}
\label{fig:cyber}
\end{figure*}

\section{Benchmark against Established Estimators}\label{sec:benchmark}
The preceding sections set the spectral estimators against ordinary kernel estimation and an adaptive
mixture on chosen targets. This section places them on the field's standard external yardstick, the
fifteen normal-mixture densities of Marron and Wand \cite{marronwand1992}, whose closed-form densities
make the integrated squared error exact rather than estimated. Seven estimators are compared: a
Silverman rule-of-thumb kernel; the improved Sheather-Jones plug-in bandwidth of Botev and colleagues
\cite{botev2010}, computed by the direct fixed-point implementation (see the implementation note of
Section~\ref{sec:exp}) and the strongest classical selector here; a true least-squares
cross-validation bandwidth (full sample at $n=100$ and $n=500$, selected on
a thousand-point subsample at $n=5000$ for tractability); an Abramson adaptive variable-bandwidth
kernel; a Gaussian mixture selected by BIC; the AD-Wiener estimator with the residue floor; and the
superposition of Section~\ref{sec:superpose} with the bundled Gaussian-mixture EM base, written
super~(GMM). Each density is drawn fifty times at each of three sample sizes, $n=100$, $500$, and
$5000$, and every estimator is scored by integrated squared error against the true density on a common
grid.

Tables~\ref{tab:bench100}, \ref{tab:bench500}, and \ref{tab:benchmark} give the per-density errors at
the three sample sizes and Fig.~\ref{fig:benchmark} the average ranks. The spectral advantage is real
but emerges with sample size rather than holding throughout, and this is the central finding of the
benchmark. At $n=5000$ the two spectral estimators, the superposition and the AD-Wiener filter, take
the top two average ranks, $2.17$ and $2.23$, ahead of every classical baseline including the
corrected improved Sheather-Jones bandwidth at $3.60$; the membership of that top two is stable under
either tie-ranking convention. At $n=500$ the superposition leads at
$2.83$, with the corrected Botev bandwidth second at $3.33$, the BIC mixture at $3.43$, and AD-Wiener
at $3.83$. At $n=100$ the order inverts, and more sharply than the first posted version suggested: the
corrected Botev plug-in is the best method in the field at $2.47$, ahead of least-squares
cross-validation at $3.60$ and AD-Wiener, the strongest spectral estimator there, at $3.87$, while the
superposition falls toward the bottom at $4.90$.

The crossover is a bias-variance tradeoff, and it is worth stating because it delineates where the
method helps and where it does not. The AD-Wiener estimator is, by construction, the
minimum-mean-squared-error linear filter on the empirical characteristic function: it keeps
high-frequency content wherever the data support it and discards it elsewhere, so it introduces little
bias and captures structure that a single bandwidth blurs. The price is variance. Estimating a
per-frequency filter from the data, rather than committing to one global bandwidth, adds estimation
variance, and at small samples that variance dominates. The methods that lead at $n=100$, the
corrected Botev plug-in and the cross-validation bandwidth, win precisely by committing to one global
bandwidth: a single well-chosen scalar carries far less estimation variance than a per-frequency
filter, and at a hundred points that is the right trade, because the problem is
variance-limited and there is not enough data to estimate the fine structure reliably, so committing to
it adds more variance than the bias it removes. As $n$ grows the variance penalty shrinks while the
low-bias advantage persists, the balance tips, and by $n=500$ the spectral estimators, which were
paying for structure they could not yet estimate, estimate it well and take the top ranks. The
superposition is the most sample-hungry of the three, because it must also fit a Gaussian-mixture base;
at $n=100$ the base cannot be estimated reliably and the superposition inherits its small-sample
variance rather than improving on it, while the band-limited residual has too little signal to separate
scales. AD-Wiener alone, which fits no mixture, is therefore the robust spectral choice at small
samples, and the superposition is the estimator to use once the sample is large enough to fit its base,
which on these densities is by $n=500$.

At the large sample the division of labor is the one the synthetic battery predicted, with one
correction in the classical field's favor. The BIC mixture leads or ties for the lead on the low-order densities whose true
form it matches exactly, the Gaussian, the bimodals, the outlier, and the double claw, because the
Marron-Wand densities are themselves normal mixtures and the mixture is fitting the correct model
class. AD-Wiener leads or ties for the lead on the strongly skewed, the kurtotic, the trimodal, and
the claw, and the corrected Botev bandwidth is genuinely strong on asymmetric multi-scale structure,
taking the two asymmetric claws outright. The superposition collects the mixture's and the band-limited estimator's
strengths, matching the better of its ingredients on nearly every density, and surpasses the BIC
mixture overall even though the benchmark hands the mixture the correct model class, because the
residual recovers the multi-scale structure the mixture omits.

Three limitations are stated plainly. The benchmark is a family of normal mixtures, which favors the
mixture base; a non-mixture stress test would isolate the spectral contribution further. The comparison
set, while it includes the strongest plug-in and cross-validation bandwidths and an adaptive kernel,
omits the log-spline and taut-string estimators, which are also strong on multi-scale densities and
remain to be added. And the superposition is not the estimator of choice at small samples: below
roughly $n=500$ on these densities the corrected Botev bandwidth or AD-Wiener alone is preferable,
since the mixture base needs data to fit. At moderate and large samples, within the set tested, the
superposition attains the best average rank among established estimators.

\begin{table*}[t]\centering
\caption{Integrated squared error ($\times10^{3}$) on the fifteen Marron-Wand densities at $n=100$, mean over fifty replications; lower is better, row minimum in bold. GMM-BIC is a Gaussian mixture selected by BIC; super~(GMM) is the superposition of Section~\ref{sec:superpose} with that mixture as its smooth base. Ties at the reported precision share the bold and the win count; average ranks use average-tie ranking.}
\label{tab:bench100}
\setlength{\tabcolsep}{5pt}\footnotesize
\begin{tabular}{l rrrrr rr}
\toprule
density & Silverman & ISJ/Botev & LSCV & Abramson & GMM-BIC & AD-Wiener & super (GMM) \\
\midrule
Gaussian            & 5.83 & 6.06 & 10.53 & 7.24 & \textbf{2.64} & 9.98 & \textbf{2.64} \\
Skewed unimodal     & \textbf{9.05} & 9.76 & 13.38 & 9.57 & 12.18 & 11.15 & 13.30 \\
Strongly skewed     & 142.2 & 53.21 & \textbf{48.64} & 112.5 & 87.07 & 53.24 & 53.45 \\
Kurtotic unimodal   & 101.7 & 55.95 & \textbf{52.40} & 58.83 & 79.23 & 55.80 & 56.22 \\
Outlier             & 61.34 & 66.06 & 85.01 & 71.00 & 45.32 & 56.71 & \textbf{31.90} \\
Bimodal             & 8.73 & 9.14 & 11.02 & \textbf{8.68} & 16.96 & 12.46 & 17.54 \\
Separated bimodal   & 46.74 & 12.87 & 14.94 & 40.26 & \textbf{9.12} & 16.15 & 16.34 \\
Skewed bimodal      & 12.57 & \textbf{12.54} & 12.99 & 12.75 & 22.89 & 13.61 & 21.21 \\
Trimodal            & 11.16 & \textbf{10.67} & 12.69 & 11.33 & 16.47 & 13.59 & 16.11 \\
Claw                & 53.17 & 49.00 & \textbf{42.47} & 54.59 & 57.88 & 45.19 & 55.26 \\
Double claw         & \textbf{10.21} & 10.51 & 12.10 & 10.35 & 20.10 & 13.56 & 18.70 \\
Asymmetric claw     & 27.81 & 27.26 & 28.18 & 28.39 & 33.50 & \textbf{26.30} & 31.43 \\
Asym.\ double claw  & 14.42 & 14.22 & 17.54 & \textbf{14.07} & 21.44 & 17.68 & 21.92 \\
Smooth comb         & 85.06 & 39.76 & \textbf{39.10} & 79.85 & 45.18 & 41.65 & 41.91 \\
Discrete comb       & 113.5 & \textbf{38.36} & 39.15 & 113.5 & 39.55 & 40.69 & 40.55 \\
\midrule
avg.\ rank      & 3.97 & \textbf{2.47} & 3.60 & 4.17 & 5.03 & 3.87 & 4.90 \\
wins (of 15)    & 2 & 3 & \textbf{4} & 2 & 2 & 1 & 2 \\

\bottomrule
\end{tabular}
\end{table*}

\begin{table*}[t]\centering
\caption{Integrated squared error ($\times10^{3}$) on the fifteen Marron-Wand densities at $n=500$, mean over fifty replications; lower is better, row minimum in bold. Columns as in Table~\ref{tab:bench100}. Ties at the reported precision share the bold and the win count; average ranks use average-tie ranking.}
\label{tab:bench500}
\setlength{\tabcolsep}{5pt}\footnotesize
\begin{tabular}{l rrrrr rr}
\toprule
density & Silverman & ISJ/Botev & LSCV & Abramson & GMM-BIC & AD-Wiener & super (GMM) \\
\midrule
Gaussian            & 1.86 & 1.88 & 2.54 & 2.24 & \textbf{0.48} & 2.37 & \textbf{0.48} \\
Skewed unimodal     & 2.85 & 2.96 & 4.14 & 2.60 & 2.33 & 2.74 & \textbf{1.86} \\
Strongly skewed     & 106.3 & 16.38 & \textbf{15.69} & 73.70 & 42.54 & 16.71 & 16.59 \\
Kurtotic unimodal   & 56.23 & 13.59 & 13.96 & 16.31 & \textbf{8.28} & 12.47 & 12.61 \\
Outlier             & 19.62 & 20.00 & 23.53 & 21.54 & 6.17 & 15.45 & \textbf{6.15} \\
Bimodal             & 3.01 & 2.52 & 3.16 & 2.19 & \textbf{1.61} & 3.16 & 1.84 \\
Separated bimodal   & 20.86 & 3.55 & 3.97 & 13.02 & \textbf{1.66} & 3.69 & 3.75 \\
Skewed bimodal      & 5.20 & 3.86 & 4.09 & 4.03 & 7.76 & \textbf{3.73} & 3.86 \\
Trimodal            & 4.81 & \textbf{3.36} & 3.62 & 3.71 & 5.42 & 3.96 & 5.67 \\
Claw                & 45.95 & 12.72 & \textbf{12.49} & 46.03 & 49.20 & 14.21 & 42.47 \\
Double claw         & 4.43 & 3.95 & 4.32 & 3.60 & \textbf{3.07} & 4.52 & 3.32 \\
Asymmetric claw     & 20.82 & 11.51 & \textbf{11.03} & 20.12 & 20.49 & 11.70 & 19.71 \\
Asym.\ double claw  & 7.87 & 6.94 & 7.28 & 6.76 & \textbf{6.08} & 7.42 & 6.23 \\
Smooth comb         & 63.54 & 17.54 & \textbf{16.40} & 58.54 & 20.36 & 17.43 & 17.44 \\
Discrete comb       & 88.28 & 18.96 & \textbf{16.06} & 83.11 & 22.44 & 16.43 & 16.20 \\
\midrule
avg.\ rank      & 5.87 & 3.33 & 3.97 & 4.73 & 3.43 & 3.83 & \textbf{2.83} \\
wins (of 15)    & 0 & 1 & 5 & 0 & \textbf{6} & 1 & 3 \\

\bottomrule
\end{tabular}
\end{table*}

\begin{table*}[t]\centering
\caption{Integrated squared error ($\times10^{3}$) on the fifteen Marron-Wand densities at $n=5000$, mean over fifty replications; lower is better, row minimum in bold. Columns as in Table~\ref{tab:bench100}. Ties at the reported precision share the bold and the win count; average ranks use average-tie ranking.}
\label{tab:benchmark}
\setlength{\tabcolsep}{5pt}\footnotesize
\begin{tabular}{l rrrrr rr}
\toprule
density & Silverman & ISJ/Botev & LSCV & Abramson & GMM-BIC & AD-Wiener & super (GMM) \\
\midrule
Gaussian            & 0.31 & 0.32 & 0.47 & 0.37 & \textbf{0.06} & 0.28 & \textbf{0.06} \\
Skewed unimodal     & 0.48 & 0.48 & 0.72 & 0.51 & 0.46 & 0.37 & \textbf{0.24} \\
Strongly skewed     & 58.90 & 2.59 & 3.96 & 30.74 & 9.42 & \textbf{2.32} & \textbf{2.32} \\
Kurtotic unimodal   & 18.41 & 2.25 & 3.04 & 1.88 & 2.12 & \textbf{1.63} & \textbf{1.63} \\
Outlier             & 2.89 & 2.91 & 4.17 & 3.28 & \textbf{0.62} & 1.55 & \textbf{0.62} \\
Bimodal             & 0.60 & 0.43 & 0.61 & 0.38 & \textbf{0.15} & 0.38 & 0.16 \\
Separated bimodal   & 5.07 & 0.59 & 0.83 & 1.49 & \textbf{0.15} & 0.46 & \textbf{0.15} \\
Skewed bimodal      & 1.18 & 0.56 & 0.82 & 0.49 & 1.29 & 0.44 & \textbf{0.18} \\
Trimodal            & 1.35 & 0.57 & 0.89 & 0.75 & 1.28 & \textbf{0.53} & 1.29 \\
Claw                & 32.93 & 1.98 & 2.69 & 26.42 & 2.35 & \textbf{1.50} & \textbf{1.50} \\
Double claw         & 2.09 & 1.85 & 2.06 & 1.85 & \textbf{1.66} & 1.85 & \textbf{1.66} \\
Asymmetric claw     & 12.51 & \textbf{2.51} & 3.49 & 10.73 & 7.76 & 2.57 & 6.17 \\
Asym.\ double claw  & 4.74 & \textbf{2.01} & 2.77 & 4.23 & 4.58 & 2.32 & 4.59 \\
Smooth comb         & 40.79 & 3.89 & 7.57 & 37.55 & 11.24 & 3.80 & \textbf{3.78} \\
Discrete comb       & 47.65 & 2.74 & 7.46 & 39.19 & 14.92 & 2.43 & \textbf{2.42} \\
\midrule
avg.\ rank      & 6.27 & 3.60 & 5.13 & 4.93 & 3.67 & 2.23 & \textbf{2.17} \\
wins (of 15)    & 0 & 2 & 0 & 0 & 5 & 4 & \textbf{11} \\

\bottomrule
\end{tabular}
\end{table*}

\begin{figure}[t]\centering
\includegraphics[width=0.96\columnwidth]{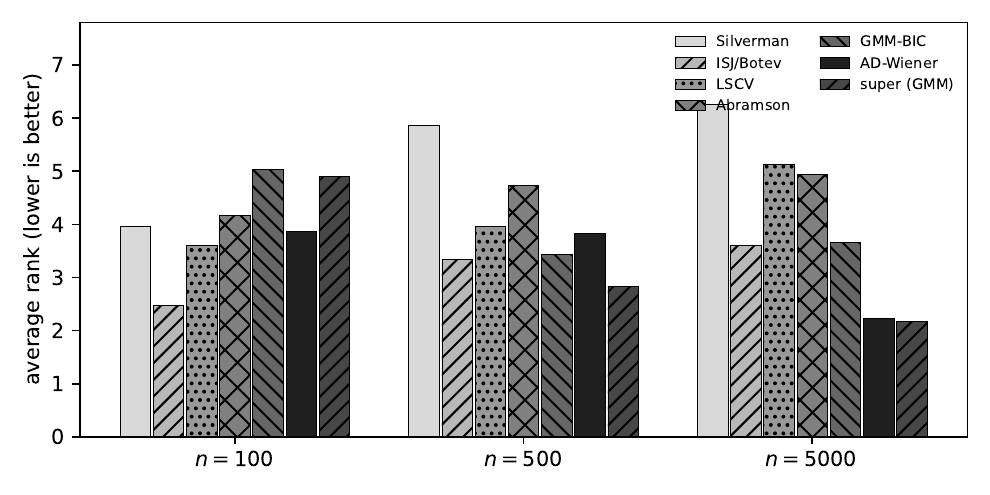}
\caption{Average rank over the fifteen Marron-Wand densities, lower is better, at $n=100$, $500$, and
$5000$. The spectral advantage emerges with sample size: at $n=5000$ AD-Wiener and the superposition
take the top two ranks ahead of every classical baseline, while at $n=100$ the corrected Botev plug-in
leads the field and the superposition trails.}
\label{fig:benchmark}
\end{figure}

\section{Software, Reproducibility and Data Availability}\label{sec:software}
The estimators and experiments are released as a reproducible bundle, available at
\url{https://github.com/mitch-thornton/kde-ad-wiener}. The mixture and group
extensions are organized behind a small plugin interface with three hooks: a candidate group
library, a blind group-matching backend, and an adaptive mixture estimator. Each ships a
self-contained default, the cyclic group, the cyclic matched group, and a Gaussian-mixture
estimator with model-selected order, so every figure and table reproduces with no external
dependency. Every mixture result reported here uses this bundled BIC Gaussian-mixture EM. A
practitioner may register a richer group library, a blind group-matching engine, or an alternative
mixture estimator in place of the defaults. The real-data study of Section~\ref{sec:crsp} is likewise scripted end to end: an extraction
step writes the CRSP series into a \texttt{data/} directory and an analysis step writes the
risk, Value-at-Risk, and tail-risk results into a \texttt{results/} directory as JSON, separating
the licensed-data step from the dependency-free analysis. The heaping study of
Section~\ref{sec:heaping} is scripted the same way against the public NHANES files, which the
analysis step reads directly to regenerate both figures. The particle-physics and galaxy-redshift
studies of Sections~\ref{sec:cern} and~\ref{sec:sdss} are scripted likewise, each reading a public
released or exported CSV directly to regenerate its figure and held-out table.

The CRSP daily stock data of Section~\ref{sec:crsp} are proprietary and available to subscribers
through Wharton Research Data Services \cite{crsp}; the extraction script and the derived per-series
results that reproduce the figures are included in the bundle. The NHANES 2017--2018 files of
Section~\ref{sec:heaping} are public \cite{nhanes2018}; the analysis script regenerates both figures
from them. The CMS dimuon data of Section~\ref{sec:cern} are public (CERN Open Data, CC0)
\cite{cms_dimuon2017}; the analysis script reads the released CSV directly. The SDSS DR18 galaxy
redshifts of Section~\ref{sec:sdss} are public \cite{sdss_dr18}, retrieved through the SkyServer SQL
interface \cite{sdss_skyserver}; the analysis script reads the exported CSV directly.


\section{Conclusion}\label{sec:conc}
Kernel density estimation is recast by reading bandwidth selection as a
spectral-support problem in the characteristic-function domain: the binned data's
group-averaged spectrum is the squared empirical characteristic function, and the bandwidth
is the cutoff where it meets the $1/n$ floor. The resulting selector matches the rule of
thumb on smooth densities and approaches the best fixed bandwidth on hard ones. Going beyond
a fixed kernel, the per-frequency optimal taper gives an adaptive estimator that surpasses
the best fixed bandwidth on most standard densities, including sharply peaked and comb-like
cases on which fixed-bandwidth methods fail, and it extends in the same domain to deconvolution under
known measurement error. Because that estimator resolves sharp structure but does not fit a smooth base
as economically as a mixture, a Gaussian mixture is combined with it as a piecewise partition and as a
superposition of a smooth base and a band-limited residual, the superposition being made the default.
On the standard Marron-Wand benchmark, scored by exact integrated squared error across three sample
sizes, the advantage emerges with sample size through a bias-variance tradeoff: the spectral estimators
carry low bias but pay in variance, so a corrected implementation of the Botev plug-in bandwidth leads
at $n=100$ and AD-Wiener alone
is the robust choice, while at $n=5000$ the Wiener filter and the superposition take the top two average
ranks ahead of every classical baseline including the corrected Botev diffusion bandwidth, the
superposition surpassing even a Gaussian mixture handed the correct model class. The same spectrum supplies a data-driven
floor robust to heaped data, and the estimators are validated on six real datasets, recovering tail risk
a Gaussian fit understates on CRSP returns, staying robust to rounding on NHANES self-reports, improving
held-out likelihood on the multi-scale CMS dimuon and SDSS galaxy-redshift spectra, recovering the
uniform target of a verified NIST randomness-beacon stream almost exactly, and giving the lowest
held-out likelihood on UNSW-NB15 network-traffic features, with a synthetic-data quality-checking
use-case recovering contrived target densities. All claims are validated numerically and the experiments
are reproducible.


\bibliographystyle{IEEEtran}
\bibliography{refs}

\end{document}